\documentclass[12pt,a4paper]{article}
\pdfoutput=1
\usepackage[utf8]{inputenc}
\usepackage[english]{babel}
\usepackage{amsmath}
\usepackage{amsfonts}
\usepackage{amssymb}
\usepackage{amsthm}
\usepackage{graphicx}
\usepackage{stmaryrd}
\usepackage[titletoc]{appendix}
\usepackage[final]{pdfpages}
\author{Kurt Pagani}
\title{Transference Plans and Uncertainty}
\date{{\small October 2015}}

\newcommand{\norm}[1]{\left\lVert#1\right\rVert}
\newcommand{\RR}[1]{\mathbb{R}^#1}
\newcommand{\CC}[1]{\mathbb{C}^#1}

\newcommand{\LL}[1]{L^2(\RR{#1})}

\newcommand{\Prob}[1]{\mathcal{P}({\RR #1})}
\newcommand{\Probc}[1]{\mathcal{P}({\RR #1 \times \RR #1})}
\newcommand{\HH}{\mathcal{H}}
\newcommand{\TT}{\mathcal{T}}
\newcommand{\VV}{\mathcal{V}}
\newtheorem{theorem}{Theorem}
\newtheorem{lemma}{Lemma}
\newtheorem{prop}{Proposition}
\newtheorem{definition}{Definition}
\newtheorem{corollary}{Corollary}
\newtheorem{remark}{Remark}
\newcommand{\assign}{:=}
\newcommand{\tmem}[1]{{\em #1\/}}
\newcommand{\tmop}[1]{\ensuremath{\operatorname{#1}}}
\newcommand{\tmverbatim}[1]{{\ttfamily{#1}}}
\begin{document}
\maketitle
\begin{abstract}
We will discuss methods of {\sl Optimal Transportation Theory} and its relations
to problems in quantum mechanics. This essentially means that the cost function
is some Hamiltonian $\HH(q,p)$ on phase space (symplectic manifold), and the marginal 
measures that have to be transported are linked by a (implicit) transformation group.   
\end{abstract}
\tableofcontents
\vspace{0.5cm}
This article is a (mostly non-technical) overview, showing the main
connections and some unsettled questions. To keep the size within reasonable bounds,
we can only scratch the surface, that is we will consider the {\sl flat} case and the 
unitary Fourier transformation as generator while omitting time dependence at all.
Since we only have to deal with locally compact spaces (actually, more often than not
${\RR d}$), {\sl measures} will be understood by the functional approach.    

Transportation theory (see e.g. \cite{Ambrosio2013},\cite{villani_topics_2003}) is about
- among other things - {\sl transference plans}, i.e. probability measures $\gamma$ on 
a product space $X\times Y$, such that the {\sl marginal measures}
\footnote{That is $\pi_{X,\#}\mu$ and $\pi_{Y,\#}\mu$, the push-forwards by the projection
maps} 
on $X$ and $Y$ respectively coincide with two fixed measures $\mu$ and $\nu$, which in 
turn are the objects that shall be transported by a so called transport map 
$T:X\longrightarrow Y$
by {\sl push-forward}: $T_\#\mu=\nu$. Typical questions: 
{\sl is there a mapping} $p=T(q)$ {\sl such that}
\begin{equation*}
     \int_{X\times Y} \HH(q,p)\,d\gamma(q,p) =  \int_{X\times Y} \HH(q,T(q))\,d\gamma(q,p)=
      \int_{X} \HH(q,T(q))\,d\mu(q),
\end{equation*} 
i.e. $d\gamma(q,p)=d\mu(q)\,\delta_{T(q)}(p)$, or: {\sl what can we say about}
\begin{equation*}
    \sup_{\gamma\in\Gamma(\mu,\nu)} \int_{X\times Y} \chi_{\Lambda}(q,p)\,d\gamma(q,p),
\end{equation*}
for a compact subset $\Lambda$? There is a wealth of new ideas and methods which will 
remain unmentioned here, but may be easily adapted from the 
excellent book \cite{villani2008optimal}.
\newpage
\section{Overview and Notation}
In the sequel we are going to use some results and terminology of
{\sl mass transportation theory}, where \cite{villani_topics_2003} serves
as the main reference.

Let $\Prob{n}$ denote the space of probability measures on $\RR n$ and for
$\varphi\in\LL{n}$ let $\hat\varphi$ denote its (unitary) Fourier transform.
Each normalized $\varphi\in\LL{n}$ gives rise to a measure
\begin{displaymath}
   \nu_{\varphi}(f)=\int_{\RR n} f(x) |\varphi(x)|^2 \,dx,
\end{displaymath}
where $f\in C_{0}(\RR n)$, the continuous functions with compact support. Then
we define the mapping
\begin{equation}
   \mu : \LL{n}\cap\{||\varphi||=1\} \longrightarrow \Probc{n} 
\end{equation}
\begin{displaymath}
   \varphi \longmapsto \nu_{\varphi}\otimes\nu_{\hat\varphi},
\end{displaymath}
that means $\mu_{\varphi}$ is the (unique) product measure with marginals 
$\nu_{\varphi}$ and $\nu_{\hat\varphi}$. Furthermore we denote by $\Gamma(\varphi)$
the subset of $\Probc{n}$ whose elements have the aforementioned marginals.

Let $\HH:\RR{n}\times\RR{n}\longrightarrow\mathbb{R}$ be lower semi-continuous
and bounded below, then we call
\begin{equation}
       K_\HH(\varphi) = \inf_{\gamma\in\Gamma(\varphi)} 
          \int_{\RR{n} \times \RR{n}} \HH(q,p) \,d\gamma(q,p)
\end{equation}\label{KantEnergy}
the {\sl Kantorovich} energy of $\varphi$. Similarly we call
\begin{equation}
      E_\HH(\varphi) = \int_{\RR{n} \times \RR{n}} \HH(q,p) \, d\mu_{\varphi}(q,p)
\end{equation}\label{SchroedEnergy}
the {\sl Schr\"odinger} energy, for reasons that will be enlightened soon. 
Monge's formulation of the {\sl optimal transport problem} reads in our case:
%a map $T:{\RR n}\rightarrow \RR{n}$ which minimizes the functional
\begin{equation}
   M_\HH(\varphi)=\inf_{T} \left\{\int_{\RR n} \HH(q,T(q) d\nu_{\varphi} : 
       T_\# \nu_{\varphi} = \nu_{\hat\varphi} \right\},
\end{equation}
which means to find a minimizing map $T:{\RR n}\rightarrow \RR{n}$, that 
transports the measure $\nu_\varphi$ to $\nu_{\hat\varphi}$ by pushing forward
\footnote{$T_\#\nu(f)=\nu(T^\#(f)$, where $T^\#(f)=f\circ T$ denotes {\sl pull-back}}:
\begin{displaymath}
     T_\# \nu_{\varphi}(f):=\nu_{\varphi}(f\circ T)=
        \int_{\RR n} f(T(q)) |\varphi(q)|^2 \,dq 
          = \int_{\RR n} f(p) |\hat\varphi(p)|^2 \,dp.
\end{displaymath}
If all quantities involved were smooth enough and $T$ one to one, then we would get 
the condition
\begin{equation}\label{MAEQ}
     |\varphi(q)|^2 = |\hat\varphi(T(q))|^2\, |\det{DT(q)}|
\end{equation}
by a simple change of coordinates.
\subsection{Schrödinger Energy}
Suppose $\HH$ has the familiar form $\HH(x,k)= \frac{\hbar^2}{2m}|k|^2+\VV(x)$,
then we easily calculate that 
\begin{displaymath}
     K_\HH(\varphi) = E_\HH(\varphi)=
        \frac{\hbar^2}{2m} \int_{\RR n} |k|^2 |\hat\varphi(k)|^2  dk +   
        \int_{\RR n} \VV(x) |\varphi(x)|^2  dx
\end{displaymath}
holds. Furthermore, if $\partial_j\varphi\in L^2(\RR n)$,  
then the above expression reduces to 
\begin{equation}
          K_\HH(\varphi) = E_\HH(\varphi)= \int_{\RR n}
              \left(  \frac{\hbar^2}{2m} |\nabla\varphi(x)|^2
              + \VV(x) |\varphi(x)|^2 \right) dx.
\end{equation}

Whether the energy is finite or not will also depend on the behaviour of $\VV$, of
course. 
In a similar way the above deduction holds whenever the {\sl cost function} $\HH$
has the form $\HH(q,p)=\mathcal{T}(p)+\mathcal{V}(q)$, that is the 
{\sl Kantorovich energy} coincides with $E_\HH$ which in turn means that the 
{\sl transference plan} $\mu_\varphi=\nu_{\varphi}\otimes\nu_{\hat\varphi}$ 
is optimal. Villani notes with reference to the sand pile example \cite{villani_topics_2003} 
\begin{verse}
    \ldots this corresponds to the most stupid transportation plan that one may 
    imagine: any piece of sand, regardless of its location, is distributed
    over the entire hole, proportionally to the depth.
\end{verse} 
He certainly would not claim that quantum mechanics were stupid, however, we 
recognize that the procedure mentioned is just another formulation of the
uncertainty principle (replacing sand pile/hole by position/momentum, although this
analogy should not be taken too serious). This is in
strong contrast to the corresponding {\sl Monge} problem 
(omitting the factor $\hbar^2/2m$ from now on),
\begin{displaymath}
       M_\HH(\varphi)=\inf_{T} \left\{\int_{\RR n}
           \left( |T(x)|^2+\VV(x) \right) d\nu_{\varphi} : 
       T_\# \nu_{\varphi} = \nu_{\hat\varphi}\right\},
\end{displaymath}
where, since $T$ is a map, there is no such distribution ({\sl mass} 
cannot be split by Monge transport). Although we speak here of {\sl virtual}
transport, the analogies are sometimes useful, in that $d\gamma(x,k)$
measures the amount of {\sl mass} transferred from $x$ to $k$. Therefore,
a general $\gamma$ may {\sl smear out} $x$ (à la multi-valued mappings), whereas 
a transference plan of the form $(id\times T)_\#\mu$ cannot. On the other hand,
assume $\gamma = \mu\otimes \delta_\kappa$ (an extreme case which is of no concern
in this paper), yields $\int \HH(x,k) d\gamma(x,k) = \int \HH(x,\kappa) d\mu(x)$,
this means, everything will be transported to $\kappa$. Such pathologies are
excluded for $\gamma\in\Gamma(\varphi)$, of course.

By a theorem of Brenier-McCann (\cite{mccann_existence_1995}, see Appendix Thm. \ref{mcann_thm}),
there is a convex function $\phi$ on $\RR n$ such that
\begin{displaymath}
                  (\nabla\phi)_{\#}\nu_\varphi = \nu_{\hat\varphi},
\end{displaymath}
whence we have \footnote{notice that $K_\HH$ is a relaxation of $M_\HH$ since 
an admissible transport map $T$ always gives rise to a transference plan
$(id\times T)_{\#} \nu_{\varphi}\in\Gamma(\varphi)$} 
\begin{displaymath}
       K_\HH(\varphi)\leq M_\HH(\varphi) \leq \int_{\RR n}
           \left( |\nabla\phi(x)|^2+\VV(x) \right) |\varphi(x)|^2 dx = E_\HH(\varphi).
\end{displaymath}
Therefore, all three quantities coincide in case of $\HH(x,k)=\TT(k)+\VV(x)$. This
is no surprise because $\int (\TT+\VV) d\gamma$ is constant on $\Gamma(\varphi)$.
Now, if we supposed for the moment the existence of a ground state $\varphi_0>0$
to $E_\HH$ (more precisely to the self-adjoint operator corresponding to $\HH$),
we would find the identities
\begin{displaymath}
    E_\HH(\varphi_0)=K_\HH(\varphi_0)= 
       \int_{\RR n}
           \left( |\nabla\log{\varphi_0(x)}|^2+\VV(x) \right) |\varphi_0(x)|^2 dx
           = M_\HH(\varphi_0). 
\end{displaymath}
This leads to the question: 
\begin{verse}
     Can $\nabla\phi = - \nabla\log{\varphi_0}$ be a {\sl Brenier} map?
\end{verse}
In the first place $\phi=- \log{\varphi_0}$ is required to be convex, or equivalently,
the ground state $\varphi_0(x)=C e^{-\phi(x)}$ should be {\tt log-}concave,
a property that is not uncommon for certain potentials $V$. A far more
stringent condition, however, is the requirement 
$(-\nabla\log{\varphi_0})_{\#}\nu_{\varphi_0} = \nu_{\hat\varphi_0}$, which,
assuming some smoothness and recalling $(\ref{MAEQ})$, reads as
\begin{equation}
           |\varphi_0(q)|^2 = |\hat\varphi_0(\nabla\phi(q))|^2\, |\det{D^2\phi(q)}|.
\end{equation}
Actually, the ground state of the harmonic oscillator  
$\varphi_{ho}(x)=C e^{-\frac{1}{2}|x|^2 }$ satisfies the above equation and 
consequently in this particular case
$T(x)=\nabla\phi(x) = -\nabla\log{\varphi_{ho}(x)}=x$ is a transport map.
Are there others? Probably not, but we do not know. 
\subsection{General $\HH$}
If the Hamilton function $\HH$ does not split up as above, then we only have
$K_\HH(\varphi)\leq E_\HH(\varphi)$ and since the infimum in (\ref{KantEnergy})
is always attained \footnote{under the conditions given at the beginning} , there is a $\gamma_{\varphi}\in\Gamma(\varphi)$ such that
$\gamma_{\varphi}(\HH) \leq \mu_{\varphi}(\HH)$. In the following let us denote
by $\Gamma_n = {\RR n}\oplus {\RR n}$ a $2n$-dimensional phase space, where
there should be no confusion among the meanings of $\Gamma$, e.g. we 
have $\Gamma(\varphi)\subset{\cal P}(\Gamma_n)$.
Whether the minimization problem
\begin{equation}\label{eigenval0}
                 \lambda_0 = \inf \left\{ \int_{\Gamma_n} \HH d\mu_\varphi: 
                       \varphi\in L^2(\RR n), ||\varphi||_2 = 1  \right\}
\end{equation}
has a solution will depend on the function $\HH$ under
consideration, and even if there is a solution, it is by no means granted that
it will be a ground state of a corresponding self-adjoint {\sl Hamiltonian}.
Existence questions will not be our concern at this point, therefore we will
take the existence of a minimizer $\varphi_0\in L^2(\RR n)$ for granted. Since
we have assumed the function $\HH$ to be bounded below (and l.s.c) it is obvious
that 
\begin{displaymath}
                  \lambda_0 \geq \inf_{\Gamma_n} \HH > - \infty
\end{displaymath} 
and moreover it holds that $\lambda_0=E_\HH(\varphi_0)\geq K_\HH(\varphi_0)$. When we define (assuming $\HH$ fixed) 
\begin{displaymath}
        F_{\varphi}(x) = \int_{\RR n} \HH(x,k)\, |\hat\varphi(k)|^2 dk
\end{displaymath} 
as well
\begin{displaymath}
        G_{\varphi}(k) = \int_{\RR n} \HH(x,k)\, |\varphi(x)|^2 dx
\end{displaymath}
and  recall that $\mu_{\varphi}=\nu_{\varphi}\otimes\nu_{\hat\varphi}$ holds by
definition, we obtain 
\begin{displaymath}
        E_\HH(\varphi)=\int_{\Gamma_n} \HH(x,k)\, d\mu_{\varphi}(x,k) =
           \int_{\RR n} F_{\varphi}(x) d\nu_{\varphi}(x) =
           \int_{\RR n} G_{\varphi}(k) d\nu_{\hat\varphi}(k). 
\end{displaymath} 
Now we may state the {\sl Euler equations} which a minimizer must satisfy.
\begin{prop}
Let $\varphi_0\in L^2(\RR n)$ be a critical point of $E_\HH(\varphi)$, then
it satisfies the equation (in ${\cal D}'(\RR n)$)
\begin{equation}\label{EulerEq}
     \left(2 E_0 - F_{\varphi_0}(x)\right) \varphi_0(x) = 
         \int_{\RR n} G_{\varphi_0}(k)\, \hat\varphi_0(k)\, 
            e^{i\langle k,x \rangle}\, dm_n(k),
\end{equation} 
where $E_0 = E_\HH(\varphi_0)$ and $dm_n(k):=(2\pi)^{-\frac{n}{2}}\,dk$.
\end{prop}

Whether (\ref{EulerEq}) is valid almost everywhere w.r.t. Lebesgue measure 
depends (here again) on the function $\HH$. The inverse Fourier transform on
the right hand side should be understood symbolically, unless 
$G_{\varphi_0}\hat{\varphi_0}\in L^1(\RR n)$. In a compact notation the
equation for a critical point of $E_\HH$ is 
\begin{displaymath}
  F_{\varphi}\,\varphi+ {\breve{G}_{\varphi}} \star\varphi = 2\lambda\,\varphi                     
\end{displaymath}
where the convolution is defined here as 
$(f\star g)(x)=\int_{\RR n} f(x-y)\,g(y)\,dm_n(y)$. It is easily checked that
in case of $\HH(x,k)=|k|^2+\VV(x)$, (\ref{EulerEq}) reduces to 
$(-\Delta+\VV(x)) \varphi_0=E_0 \varphi_0$. Our main interest, however, is
$\HH$ being the indicator function of an open subset of $\Gamma_n$ which is
obviously bounded, measurable and lower semi-continuous. This is, as will
be outlined further below, connected to the question:
\begin{verse}
    How big can we make 
    \begin{equation}
             \int_{\Lambda} |\varphi(x)|^2 \, |\hat\varphi(k)|^2 dx\,dk,
    \end{equation}
    given a compact subset $\Lambda$ of phase space $\Gamma_n$?
\end{verse} 
Actually, the question may be posed for $\Lambda\subset\Gamma_n$ having finite
Lebesgue measure.
\subsection{Duality}
One of the corner stones of mass transportation theory certainly is 
{\sl Kantorovich's} duality formula (\cite{villani_topics_2003}, Theorem 1.3)
which, translated to our needs, says
\begin{equation}\label{duality}
   K_\HH(\varphi) = \sup_{\TT(k)+\VV(x) \leq \HH(x,k)} 
              \left\{ \int_{\RR n} \TT(k) |\hat\varphi(k)|^2 \,dk +
                               \int_{\RR n} \VV(x) |\varphi(x)|^2 \,dx                         
                                      \right\},
\end{equation}
where the functions $\TT,\VV$ may either be any bounded continuous functions on $\RR n$
or by extension $(\TT,\VV)\in L^1(\nu_{\varphi})\times L^1(\nu_{\hat\varphi})$,
satisfying the inequality $\TT+\VV\leq \HH$ point-wise in the first case and almost
everywhere (with respect to the measures) in the second case. 
We cite one other result from \cite{villani_topics_2003} which will be required
later on (a precursor of Strassen's theorem, Theorem 1.27): 
Let $U$ be a non-empty open subset of $\Gamma_n$, then
\begin{equation}\label{Strassen}
     \inf_{\gamma\in\Gamma(\varphi)} \int_{U} d\gamma =
       \sup_{A\subset{\RR n}} 
          \left\{ 
            \int_{A} |\varphi(x)|^2\,dx -
            \int_{A_U} |\hat\varphi(k)|^2\, dk : A\,\; closed 
          \right\},                                                 
\end{equation}
where $A_U:=\{k\in{\RR n}: \exists x\in A , (x,k)\notin U \}$. This result
implies, setting $\HH=\chi_U$,
\begin{displaymath}
  E_{\chi_U}(\varphi)\geq K_{\chi_U}(\varphi)=
        \sup\{ \nu_{\varphi}(A)-\nu_{\hat\varphi}(A_U)
            : A\subset{\RR n},A\,\;closed \}. 
\end{displaymath}
Note that we use the notation $\nu(A)$ and $\nu(\chi_A)$ interchangeably when
there is no danger of confusion (i.e. we identify a set with its indicator 
function).
\subsection{Symplectic transformations}\label{sympl_trafos}
Let $M:\Gamma_n \rightarrow \Gamma_n$ be a symplectic transformation, represented
by a matrix of the form (we use the same symbol)
\begin{displaymath}
          M=M^{A,B,C,D}:=\begin{bmatrix}
          A & B \\ 
          C & D
          \end{bmatrix} 
\end{displaymath}
where the $n\times n$ block matrices $A,B,C,D$ satisfy the equations:
\begin{displaymath}
         \begin{array}{c}
          A^T D - C^T B = I \\
          A^T C = C^T A \\ 
          D^T B = B^T D. 
         \end{array} 
\end{displaymath}
Then we obtain for any $f\in C_0({\Gamma_n})$:
\begin{displaymath}
      M_{\#}\mu_{\varphi}(f) = \mu_{\varphi}(f\circ M) = 
         \int_{\Gamma_n} f(A x+B k, C x + D k) d\mu_{\varphi}(x,k).
\end{displaymath}
The inverse $M^{-1}$ of $M$ is easily calculated using the symplectic
condition $M^T J M = J$ to
\begin{displaymath}
           M^{-1}=J^{-1} M^T J = \begin{bmatrix}
          D^T & -B^T \\ 
          -C^T & A^T
          \end{bmatrix}
\end{displaymath}
which implies:
\begin{displaymath}
      M_{\#}\mu_{\varphi}(f) =  \int_{\Gamma_n} f(\xi,\eta)\,
      |\varphi(D^T \xi - B^T \eta)|^2 \,
      |\hat\varphi(-C^T \xi + A^T \eta)|^2 d\xi d\eta.       
\end{displaymath}
Simple examples (e.g. $n=1$ and $\phi(x)=C \exp(-\alpha |x|)$) show that 
we cannot expect the image measure $M_{\#}\mu_{\varphi}$ being an element of some
$\Gamma(\psi)$. However, two special cases immediately spring to mind:
\begin{displaymath}
          M^{A,0,0,D}_{\#}\mu_{\varphi}(f) =  \int_{\Gamma_n} f(\xi,\eta)\,
      |\varphi(D^T \xi|^2 \,
      |\hat\varphi(A^T \eta)|^2 d\xi d\eta. 
\end{displaymath}
and
\begin{displaymath}
       M^{0,B,C,0}_{\#}\mu_{\varphi}(f) =  \int_{\Gamma_n} f(\xi,\eta)\,
      |\varphi( - B^T \eta)|^2 \,
      |\hat\varphi(-C^T \xi)|^2 d\xi d\eta.
\end{displaymath}
In the first case we have $B=C=0$, so that $A^T D = I$ by the symplectic
conditions above. The second case requires $-C^T B=I$ by the same reasoning
since $A=D=0$. Hence there are two subgroups generated by matrices of the form
\begin{displaymath}
      \begin{bmatrix}
      A & 0 \\ 
      0 & A^{-T}
      \end{bmatrix}
and 
      \begin{bmatrix}
      0 & B \\ 
     -B^{-T} & 0
      \end{bmatrix}.        
\end{displaymath}
For these, the image measures are indeed of the form $d\mu_{\psi}$. If we use
the notation $\varphi_A(x)=\varphi(A x)$ we may state:
\begin{equation}
        M^{A,0,0,A^{-T}}_{\#}\mu_{\varphi} = \mu_{\varphi_{A^{-1}}}
\end{equation}
and
\begin{equation}
       M^{0,B,-B^{-T},0}_{\#}\mu_{\varphi} = \mu_{\widehat{\varphi_{B^{-1}}}}.
\end{equation}
This follows by straightforward computation. Finally we want to mention the 
special case $B=I_n$ (the identity matrix in $\RR n$), giving M=J, thus
\begin{displaymath}
        J_{\#}\mu_{\varphi}(\HH) = \mu_{\varphi}(\HH\circ J)=\mu_{\hat\varphi}(\HH),
\end{displaymath}
which is equivalent to $E_{\HH\circ J}(\varphi)=E_\HH({\hat\varphi})$. In other
words, if $\HH$ is invariant under the canonical transformation 
$x'=k, k'=-x$ and if $\varphi_0$ is a unique positive minimum of $E_\HH$, then 
$\varphi_0(x)=C \exp(-|x|^2/2)$.  
\subsection{Orthonormal Sequences in $L^2(\RR d)$}
Let $\{\varphi_j\}_{j\in J}$ be an orthonormal sequence in $L^2(\mathbb{R})$, then
a result by H. S. Shapiro, meanwhile known as {\sl Shapiro's Umbrella Theorem},
states that if given two functions $f(x)$ and $g(k)$ in $L^2(\mathbb{R})$ such that
\begin{displaymath}
         |\varphi_j(x)| \leq |f(x)|,\,\, |\hat\varphi_j(k)|\leq |g(k)|
\end{displaymath}
for all $j\in J$ and for almost all $x,k$ in $\mathbb{R}$, then $J$ must be
finite. We refer to \cite{jaming:hal-00080455} and the references therein for background information
and more details.
Recently. E. Mallinikova (\cite{malinnikova:2010},Th. 1.2) showed the following localization property of a
orthonormal sequecne $\{\varphi_j\}_{j=1}^N$: 
\begin{equation}\label{LocProp}
                N - |A| |B| \leq \frac{3}{2} \sum_{j=1}^N 
                   \left( \sqrt{\nu_{\varphi_j}({\RR d} \backslash A)} 
                      + \sqrt{\nu_{\hat\varphi_j}({\RR d} \backslash B)}                 
                   \right)
\end{equation}
where $A,B\subset {\RR d}$ are arbitrary measurable sets with finite Lebesgue
measure (i.e. $|A|,|B|<\infty$). Remembering the definition of the Radon measures
$\nu_{\varphi}$ at the beginning, $\nu_{\varphi}({\RR d} \backslash A)$ is
just 
\begin{displaymath}
          \int_{{\RR d} \backslash A} |\varphi(x)|^2\,dx.
\end{displaymath}  
The inequality (\ref{LocProp}) immediatley leads to a quantitative version of the 
{\sl Umbrella theorem} ([EM,Th. 4]) as well as to the general inequality
\begin{equation}
      \sum_{_j=1}^N \int_{\Gamma_d} \left( |x|^p+|k|^p  \right)\, 
               d\mu_{\varphi_j}(x,k) \geq C\, N^{1+\frac{p}{2 d}},
\end{equation} 
where $C$ depends only on $p>0$ and $d$. Moreover, it is also shown that the inequality is sharp up to a multiplicative constant. 
\subsection{The Nazarov-Jaming Inrequality}
Another important result we shall need is the following inequality obtained by
Nazarov for the case $d=1$ and extended by Jaming 
\cite{jaming:hal-00120268} to $d\geq 1$.

Let $A,B\subset {\RR d}$, each having finite Lebesgue measure, then there are
positive constants \footnote{Clearly, the constants may depend on the dimension 
$d$, although we do not explicitly outline this point by notation.}
$\alpha,\beta$ and $\eta(A,B)$ such that
\begin{equation}\label{JamingNazarov}
     \nu_{\varphi}({\RR d}\backslash A)+
       \nu_{\hat\varphi}({\RR d}\backslash B) \geq \alpha e^{-\beta \eta(A,B)}
\end{equation}
holds  for all $\varphi\in L^2(\RR d),\,||\varphi||=1$. The constant $\eta$
is given by
\begin{displaymath}
    \eta(A,B) = \left\{ \begin{array}{lr} |A| |B| & : d = 1
     \\ \min(|A|\,|B|,|A|^{1/d}\, w(B), w(A)\, |B|^{1/d}) & : d \ge 1 
     \end{array} \right.
\end{displaymath}
with $w(A)$ the {\sl average width} of $A$ (see \cite{jaming:hal-00120268} 
for the precise definition).

\subsection{Scaling}
For $\lambda>0$ let $\varphi_{\lambda}(x)$ denote the scaled function 
$\lambda^\frac{n}{2} \varphi(\lambda x)$, then $||\varphi_{\lambda}||=1$ whenever 
$\varphi\in L^2(\RR n)$ and $||\varphi||=1$. The Fourier transform 
$\widehat{\varphi_{\lambda}}$ of $\varphi_{\lambda}$ is easily calculated
to be equal to $\hat{\varphi}_{1/\lambda}$, therefore 
\begin{equation}
    \mu_{\varphi_{\lambda}}(f) = \int_{\Gamma_n} f(x,k) 
       \,|\varphi_{\lambda}(x)|^2 \, |\hat{\varphi}_{1/\lambda}(k)|^2 \,
       dx\,dk,
\end{equation}
for all $f\in C_0(\Gamma_n)$. The coordinate change $\xi=\lambda x, 
\eta=k/\lambda$ yields
\begin{equation}
    \mu_{\varphi_{\lambda}}(f) = \int_{\Gamma_n} f(\frac{\xi}{\lambda},
      \lambda\,\eta) 
       \,|\varphi(\xi)|^2 \, |\hat{\varphi}(\eta)|^2 \,
       d\xi\,d\eta = \int_{\Gamma_n}  f(\frac{\xi}{\lambda},
      \lambda\,\eta) d\mu_{\varphi}(\xi,\eta).
\end{equation}

\section{Maximum Probability of Compact Sets}
\begin{definition}
Let $\Lambda$ be a closed subset of $\Gamma_n$, then we define
\begin{equation}
     e(\Lambda) = \sup \left\{ \int_{\Lambda} d\mu_{\varphi} : 
                     \varphi\in L^2(\RR n), \, ||\varphi|| = 1 \right\}.
\end{equation}
\end{definition}

\begin{lemma}
Let $A,B$ be subsets of $\RR n$ having finite Lebesgue measure, that is $|A|+|B|<\infty$, then exists a $\psi$ such that 
\begin{equation}
          \nu_{\psi}(A)+\nu_{\hat\psi}(B) = 0.
\end{equation}
\end{lemma}

\begin{proof}
This follows by Corollary 2.5.A in \cite{havin_uncertainty_2012}. Actually it is
shown that there always is a $\varphi\in L^2(\RR n)$ such that for any given pair
$g,h$ of functions in $L^2(\RR n)$ the restriction of $\varphi$ to $A$ and that
of $\hat\varphi$ to $B$ coincides with the restriction of $g$ to $A$ and
$h$ to $B$ respectively. 
\end{proof}

\begin{prop}
Let each of $A,B$ be the complement of a bounded open subset in $\RR n$, then 
\begin{displaymath}
       e(A\times B) = 1.
\end{displaymath}
\end{prop}

\begin{proof}
Set $U={\RR n}\backslash A,\, V={\RR n}\backslash B$, then for each $\varphi$ we 
have $\mu_{\varphi}(A\times B)=\nu_{\varphi}(A)\,
\nu_{\hat\varphi}(B)=(1-\nu_{\varphi}(U))\, (1-\nu_{\hat\varphi}(V))$. By 
the lemma above we may choose a $\psi$ such that 
$\nu_{\psi}(A)=\nu_{\hat\psi}(B) = 0$, thus $\mu_{\psi}(A\times B)=1$.
\end{proof}

\begin{prop}
Let $A,B$ be subsets of $\RR n$ such that $|A|+|B|<\infty$, then for every 
normalized $\varphi\in L^2(\RR n)$
\begin{equation}\label{eUbound}
   \mu_{\varphi}(\chi_{A\times B}) \leq 
      \left( 1-\frac{\alpha}{2} e^{-\beta\, \eta(A,B)} \right)^2
\end{equation}
with constants $\alpha,\beta$ and $\eta$ as in (\ref{JamingNazarov}). 
\end{prop}

\begin{proof}
Using (\ref{JamingNazarov})  we get 
$2-\left(\nu_{\varphi}(A)+\nu_{\hat\varphi}(B)\right)\geq \alpha 
  e^{-\beta\,\eta(A,B)}$. Dividing both sides by two and applying the
  arithmetic-geometric mean inequality yields 
  $ \sqrt{\nu_{\varphi}(A)\nu_{\hat\varphi}(B)}\leq 
     1-\frac{\alpha}{2} e^{-\beta\, \eta(A,B)}$, which implies (\ref{eUbound}).
\end{proof}

\begin{corollary}\label{corollary_e}
Let $\Lambda\subset\Gamma_n$ be compact, then
\begin{equation}
             e(\Lambda) \leq 
      \left( 1-\frac{\alpha}{2} e^{-\beta\, 
          \eta(\pi_1(\Lambda),\pi_2(\Lambda))} \right)^2,
\end{equation}
where $\pi_1,\pi_2:\Gamma_n \rightarrow {\RR n}$ are the standard projections
and the constants $\alpha,\beta,\eta$ are as in (\ref{JamingNazarov}). 
\end{corollary}

\begin{proof}
The images of the projections $\pi_1,\pi_2$ are again compact, thus measurable
and of finite Lebesgue measure. 
\end{proof}
\rm
If we replace compactness by finite Lebesgue measure or closed only, then
we have to deal with analytic sets. 
Since 
\begin{displaymath}
    e(\Lambda)=\sup_{\varphi} \mu_{\varphi}(\Lambda)=
      1-\inf_{\varphi} \mu_{\varphi}(\Gamma_n\backslash\Lambda)
\end{displaymath}
we have the relation to $E_\HH$ with $\HH=\chi_{\Gamma_n\backslash\Lambda}$. Since
$U=\Gamma_n\backslash\Lambda$ is open we can also apply (\ref{Strassen}). 
\subsection{Optimal Bounds}
Optimal bounds are (at the time of writing, 2015) not known. To illustrate the
difficulties one encounters when trying to find maximizers of $e(\Lambda)$, let us
consider the case where $\Lambda=\{x^2+k^2\leq R^2\}\subset\RR{2}$. Using the first
Hermite function $\psi_0(x)=\pi^{-\frac{1}{4}}\,e^{-\frac{1}{2}\,x^2}$ as a trial
function, we obtain with $\psi_0=\hat\psi_0$ in mind:
\begin{equation}\label{hermite_disk}
   e(\Lambda) \geq \frac{1}{\pi}\int_{\{x^2+k^2\leq R^2\}} e^{-(x^2+k^2)}\, dx\,dk
   = 1-e^{-R^2}.
\end{equation}
Nazarov's inequality (\ref{JamingNazarov}) for $d=1$, with 
$A=[-\frac{a}{2},\frac{a}{2}]$, $B=[-\frac{b}{2},\frac{b}{2}]$, reads with $\psi_0$:
\begin{equation*}
    \frac{1}{\sqrt{\pi}} \int_{\mathbb{R}\backslash A} e^{-x^2}\,dx +
    \frac{1}{\sqrt{\pi}} \int_{\mathbb{R}\backslash B} e^{-k^2}\,dk =
    \operatorname{erfc}({\frac{a}{2}})+\operatorname{erfc}({\frac{b}{2}})
    \geq \alpha\,e^{-\beta\, a b},
\end{equation*}
where 
\begin{align*}
    \operatorname{erfc}(x) & = 1-\operatorname{erf}(x) \\[5pt]
                           & = \frac{2}{\sqrt\pi} \int_x^\infty e^{-t^2}\, dt \\[5pt]
\end{align*} 
denotes the {\sl complementary error function}. Now, when we set $a=b=2\,R$, we get
\begin{equation}
    \operatorname{erfc}(R) \geq \frac{\alpha}{2}\,e^{-4\beta R^2}
\end{equation}
Indeed, there exist quite optimal Chernoff-type bounds (\cite{Chang2011}, Theorem 2) for 
$\operatorname{erfc}$: 
\begin{equation}
   \operatorname{erfc}(R) \geq \rho\, e^{-\sigma\,R^2}
\end{equation}
if $\sigma > 1$ and $0<\rho\leq \sqrt{\frac{2\,e}{\pi}} \frac{\sqrt{\sigma-1}}{\sigma}$.
On the other hand, there is also an upper bound (\cite{Chang2011}, Theorem 1) of the
same kind:
\begin{equation}
   \operatorname{erfc}(R) \leq \kappa\, e^{-\lambda\,R^2}
\end{equation}
provided that $\kappa \geq 1$ and $0<\lambda\leq 1$, more precisely if and only if
$\kappa,\lambda$ satisfy these relations. It is actually believed that a Gaussian
function with $A,B$ balls of radius $R$ is optimal (see introduction in \cite{jaming:hal-00120268}), however, there is no proof yet. It seems to be even
more difficult to prove optimality if $\Lambda$ is not a product, e.g. as in
(\ref{hermite_disk}). The following lines might illustrate this. 

For $Q=[-R,R]\times [-R,R]$ we get
\begin{equation*}
    e(Q) \geq \operatorname{erf}(R)^2 = (1-\operatorname{erfc}(R))^2,
\end{equation*}
which is in accordance with (\ref{eUbound}) and the bounds of $\mathtt{erfc}$
discussed above. Since $\Lambda\subset Q$ we also have that
\begin{displaymath}
  e(Q)\geq e(\Lambda) \geq 1-e^{-R^2}, 
\end{displaymath}
so that by Corollary \ref{corollary_e} the upper bound for $e(\Lambda)$ is the same
as that for $e(Q)$, what is certainly not optimal. Referring to the remark at the
end of Section \ref{sympl_trafos}, one might conjecture that $\psi_0$ is a 
minimizer of $E_{\chi_{\RR{2}\backslash\Lambda}}(\varphi)$, i.e.
\begin{displaymath}
    E_{\chi_{\RR{2}\backslash\Lambda}}(\psi_0)=1-e(\Lambda)=e^{-R^2},
\end{displaymath}    
however, it seems that $\psi_0$ is not a solution of the Euler equation 
(\ref{EulerEq}).
%The functions $F,G$ are easily calculated for $\varphi=\psi_0$, but then one
%can see how unattractive the equation to verify becomes when $\HH$ does not split up. 

\subsection{A Theorem of A. Steiner}
Uncertainty inequalities occur in various forms, very often disguised as a
localization principle. There is a nice theorem by Antonio Steiner \cite{Steiner1974},
that is probably not widely known (as the article is in German), so we cite it here
using our notation:
\begin{theorem}[Theorem 3, \cite{Steiner1974}]
Let $A,B\subset\mathbb{R}$ be two measurable sets, such that $|A|+|B|<\infty$. 
Suppose $\nu_{\varphi}(A) + \nu_{\hat\varphi}(B) > 1$, then
\begin{equation}
  |A| |B| \geq \max\{c_1,c_2\},
\end{equation}
where
\begin{displaymath}
    c_1= \frac{2\pi}{\nu_{\varphi}(A)}
   \left(\sqrt{ \nu_{\hat\varphi}(B)}-\sqrt{1-\nu_{\varphi}(A)}\right)^2  
\end{displaymath}
and
\begin{displaymath}
   c_2= \frac{2\pi}{\hat\nu_{\varphi}(B)}
   \left(\sqrt{ \nu_{\varphi}(A)}-\sqrt{1-\nu_{\hat\varphi}(B)}\right)^2. 
\end{displaymath}
\end{theorem}
Recall that $\varphi\in L^2(\mathbb{R})$ with $||\varphi||_2=1$, and
\begin{displaymath}
     \nu_{\varphi}(A)=\nu_{\varphi}(\chi_A)=\int_{A} |\varphi(x)|^2\,dx.
\end{displaymath}
We will sketch the proof here because it is elementary, but the idea is clever.
Writing the Fourier transform of $\varphi$ as
\begin{displaymath}
    \hat\varphi(k)=\frac{1}{\sqrt{2\pi}}\int_A \varphi(x)\,e^{-ikx}\, dx +
                   \frac{1}{\sqrt{2\pi}}\int_{\mathbb{R}\backslash A} 
                      \varphi(x)\,e^{-ikx}\, dx = \hat\varphi_A+\hat\varphi_{A^c},
\end{displaymath} 
we may proceed with
\begin{align*}
   \nu_{\hat\varphi}(B) & = \int_B |\hat\varphi_A+\hat\varphi_{A^c}|^2\, dk \\[5pt]
                        & \leq  \int_B (|\hat\varphi_A|+|\hat\varphi_{A^c}|)^2\, dk \\[5pt]
                        & \leq \ldots \mathrm{2\times\ Schwarz\ inequality} \\[5pt]
                        & \leq  \nu_{\varphi}(A)\frac{|A||B|}{2\pi} + (1-\nu_{\varphi}(A))
     + 2\,\left( \nu_{\varphi}(A)\frac{|A||B|}{2\pi}\right)^\frac{1}{2}\,
          \left(1-\nu_{\varphi}(A)\right)^\frac{1}{2}
\end{align*}
Setting $\xi=\sqrt{\frac{|A||B|}{2\pi}}$,\,$\alpha=\nu_{\varphi}(A)$, and 
$\beta= \nu_{\hat\varphi}(B)$ yields
\begin{displaymath}
    \alpha\xi^2 + 2\sqrt{\alpha}\sqrt{1-\alpha}\,\xi+(\alpha+\beta-1)\geq 0,
\end{displaymath}
implying
\begin{displaymath}
      \xi \geq \frac{\sqrt{\beta}-\sqrt{1-\alpha}}{\sqrt{\alpha}},
\end{displaymath}
that is
\begin{displaymath}
     |A| |B| \geq c_1.
\end{displaymath}
Reversing the roles gives $c_2$. The proof reveals why the condition $\alpha+\beta>1$
is necessary. Not only because it guarantees $\alpha>0\wedge\beta>0$, also because
$\xi\geq0$ is required. Finally, we want to point out the link to (\ref{LocProp}).
\subsection{Tightness}
Since we know that $K_\HH(\varphi)$ is attained by a $\gamma_\varphi$, a minimizing
sequence $\gamma_{\varphi_j}$ of
\begin{displaymath}
     \inf_{\varphi} K_\HH(\varphi)
\end{displaymath}
does not necessarily converge to a measure, and even if it does, it is not for sure
that it is in $\Gamma(\varphi)$ for some $\varphi\in L^2(\RR n)$. By Prokhorov’s 
theorem, however, it is sufficient to show the tightness of the sequence
$\{\gamma_{\varphi_j}\}_{j\geq 1}$ in order to get a weakly convergent subsequence,
that is, $\forall\epsilon>0,\exists K_\epsilon\subset\Gamma_n$, such that
\begin{displaymath}
        \gamma_{\varphi_j}(K_\epsilon) \geq 1-\epsilon
\end{displaymath}
holds $\forall j\in \mathbb{N}$. This is usually not trivial, but we have 
\begin{displaymath}
   e(A_\epsilon\times B_\epsilon)\geq  \gamma_{\varphi_j}(A_\epsilon\times B_\epsilon) 
   \geq \nu_{\varphi_j}(A_\epsilon)\,\nu_{\hat\varphi_j}(B_\epsilon),        
\end{displaymath}
that is, in view of the weak compactness of the unit ball in $L^2$, the problem
may be often reduced to merely consider the marginal measures. If the function
$\HH$ is {\sl inf-compact}, that is if the level sets $\{(x,k):\HH(x,k)\leq r\}$ are compact
for all $r\in\mathbb{R}$, then there are (usually) standard procedures to verify
tightness. For instance, if we take the additional assumption
\begin{displaymath}
     \int_{\Gamma_n} \HH(x,k) d\gamma_{\varphi_j}(x,k) \leq C,
\end{displaymath} 
for all $j\in\mathbb{N}$, and set $K_{1/m}=\{(x,k):\HH(x,k)\leq m\}$, then for all
$\gamma\in \{\gamma_{\varphi_j}\}_{j\geq 1}$
\begin{displaymath}
   m\, \gamma(\Gamma_n\backslash K_{1/m})\leq 
      \int_{\Gamma_n\backslash K_{1/m}} \HH(x,k) d\gamma(x,k) \leq C,   
\end{displaymath}   
thus
\begin{displaymath}
     \gamma(K_{1/m}) \geq 1-\frac{C}{m}.
\end{displaymath}
Note that inf-compact functions are lower semi-continuous, thus we conclude
\begin{displaymath}
    \liminf_{j\rightarrow\infty} \int \HH d\gamma_{\varphi_{m_j}}\geq \int \HH d\gamma_\star,
\end{displaymath}
where $\gamma_\star$ is the weak limit of the subsequence.
\section{Miscellaneous} 
\subsection{Approximation}
Most problems in non-relativistic quantum mechanics are based on Hamiltonian
functions which cleave into a kinetic and a potential part, or at least, they
may be transformed into such a form. There are some well known exceptions, of
course, as soon as we consider terms like e.g. $\sqrt{(k-X(x))^2+\sigma(x)}$.
Where Pauli's equation still fits into the scheme, we leave the foundational
frame if $\HH$ does not separate into $\TT+\VV$ and/or $T\cdot V$. 
From a mathematical point of view this will not be a matter of concern.
However, the physical interpretation of the limit measures ($\gamma_\star$) of
$K_\HH$ in these cases is not quite clear (at least not to me). Nevertheless,
the duality formula (\ref{duality}) suggests the following considerations. 

For $\TT,\VV\in C_b(\RR n)$ let
\begin{equation}
  \varepsilon(\TT,\VV) = 
     \inf\{\nu_\varphi(T)+\nu_{\hat\varphi}(V): \varphi\in L^2(\RR n)\},
\end{equation} 
and
\begin{equation}
  \delta_\HH(\varphi) = 
     \sup_{\TT+\VV\leq \HH} \left(\nu_\varphi(T)+\nu_{\hat\varphi}(V)
       \right),
\end{equation} 
then we have for any admissible $\HH$:
\begin{equation}\label{approx_ineq}
    \lambda_0(\HH)\geq\inf_{L^2(\RR n)} K_\HH(\varphi)=\inf_{L^2(\RR n)}\delta_\HH(\varphi)
      \geq \sup_{\TT+\VV\leq \HH} \varepsilon(\TT,\VV),
\end{equation}
where the equality sign is by (\ref{duality}) and the last inequality is a
consequence of the {\sl max-min inequality}, that is
$\sup_X\inf_Y f(x,y)\leq\inf_Y\sup_X f(x,y)$, valid for arbitrary sets $X,Y$.
It is well known that the latter inequality may be strict, but there are cases
where equality may be proved, e.g. showing the existence of a saddle point by
methods as described in \cite{ekeland_temam}.
More will be published elsewhere.
\subsection{Special case: $\HH(x,k)=|k|^2\,|X(x)|^2$}
In OTT much is known about quadratic costs like $|x-y|^2$. Generally, powers of
a distance\footnote{OTT is well defined on Polish spaces.} 
function $d(x,y)$ play an important role, for obvious reasons when
considering actual transport of goods. We, however, want to consider the {\sl 
cost function}  $\HH(x,k)=|k|^2\,|X(x)|^2$ because we know the minimizers of
$E_{\HH}$ and because it serves as a simple model where $\gamma({\HH})$ is
not constant on $\Gamma(\varphi)$. The vector field 
$X:{\RR n}\rightarrow {\RR n}$ is assumed to be sufficiently smooth (say $C^1$)
for simplicity, then $\HH$ is certainly l.s.c and bounded below (by zero).
Thus 
\begin{displaymath}
   K_{\HH}(\varphi) = \inf_{\gamma\in\Gamma(\varphi)} \int_{{\RR n}\times {\RR n}} 
      |k|^2\,|X(x)|^2 d\gamma(x,k) \geq 0
\end{displaymath}
is attained for some $\gamma_{\varphi}\in\Gamma(\varphi)$. For $E_{\HH}(\varphi)$ we
get
\begin{displaymath}
   E_{\HH}(\varphi) =   \int_{{\RR n}\times {\RR n}} 
      |k|^2\,|X(x)|^2 d\mu_{\varphi}(x,k) = \int_{\RR n} |k|^2\, 
        d\nu_{\hat\varphi}(k) \int_{\RR n} |X(x)|^2\,d\nu_{\varphi}(x).
\end{displaymath}
Now, $|k|^2\in L^1(\nu_{\hat\varphi})$, i.e. $\nu_{\hat\varphi}$ certainly has finite
second order moments if $\varphi\in H^1(\RR n)$, otherwise the integral may be
infinite. The inequality (\ref{ineq}) in Lemma 2 (Appendix), tells us (with 
$f(t)=\frac{1}{2} t^2$):
\begin{displaymath}
    E_{\HH}(\varphi) \geq \frac{1}{4} \left(\int_{\RR n} 
         \operatorname{div}(X)\,d\nu_\varphi \right)^2,
\end{displaymath}   
where we assume $|X|\in L^2(\nu_\varphi)$, and 
$\operatorname{div}(X)\in L^1(\nu_\varphi)$. The equality sign holds if
\begin{displaymath}
    \nabla\varphi(x)+\varphi(x)\,X(x) = 0.    
\end{displaymath} 

For simplicity we proceed with the special case $X(x)=x$, that leads to    
\begin{displaymath}
      E_{\HH}(\varphi) \geq \frac{n^2}{4}
\end{displaymath}
with equality for $\varphi_0(x)=C\,\exp(-\frac{x^2}{2})$. Therefore, we get
for (\ref{eigenval0}) the value $\lambda_0=\frac{n^2}{4}$. Verifying (\ref{EulerEq})
actually shows that $\varphi_0$ is a critical point of $E_\HH(\varphi)$, with
$X(x)=x$. Consequently, $\lambda_0$ is an upper bound to $\inf K_{\HH}$ when 
recalling (\ref{approx_ineq}). Finding lower bounds is much trickier because 
approximating $|x|^2\, |k|^2$ by functions $f(x)+g(k)$ is much more unpleasant
than controlling the function $\langle x,k\rangle$ for example. While the latter
leads to the usual convex conjugates, the former requires e.g. the notion of 
{\sl c-concavity}, that is in our special case,
\begin{displaymath}
   f(x) = \inf_{k\in{\RR n}} \left( |x|^2\, |k|^2 - g(k)\right),
\end{displaymath} 
and {\sl cyclic monotonicity} (see \cite{villani_topics_2003}). Actually,
there is not much known about the measures $\gamma_{\varphi}$ nor about
optimal bounds for $K_{\HH}$. At least, however, we know that
\begin{displaymath}
   K_{\HH}(\varphi_0)= {\operatorname{max}}_{\psi\in L^1(\nu_{\varphi_0})} \,
     \int_{\RR n} [\psi^c(x) - \psi(x)]|\varphi_0(x)|^2\,dx \leq \frac{n^2}{4},
\end{displaymath} 
that is, the dual problem is attained as well 
(see Theorem 5.10 in \cite{villani2008optimal}). Therefore, $\mu_{\varphi_0}$ 
is not optimal anymore. This was to be expected because of the {\sl interaction}
of $x$ and $k$ in $\HH$. The measure $\gamma_{\varphi_0}$ now has to be supported
on some subset (actually a closed c-cyclically monotone) such that 
$\HH \approx_\gamma (\psi_0^c(k)-\psi_0(x))$.
%
%Finally, let us have a look at $M_{\HH}$
%for the case $X(x)=x$.
%\begin{displaymath}
%    M_{\HH}(\varphi) = \inf \left\{ \int_{\RR n} |x|^2\, |T(x)|^2\,d\nu_{\varphi}(x):
%        T_{\#}\nu_{\varphi}=\nu_{\hat\varphi}   \right\}.
% \end{displaymath}
% 
%Consequently, we obtain
%\begin{displaymath}
% \inf_{\varphi\in L^2(\RR n)}\, \inf_{\gamma\in\Gamma(\varphi)} 
%    \int_{{\RR n}\times {\RR n}} 
%      |k|^2\,|x|^2 d\gamma(x,k) \leq \frac{n^2}{4}.
%\end{displaymath} 

%
\subsection{Case $n=1$}
We will see that even the one-dimensional case (i.e. phase space $\RR 2$) may
be quite involved. First, let us consider the distribution functions $F,G$ of 
the measures $\nu_\varphi$ and $\nu_{\hat\varphi}$ respectively:
\begin{equation*}
    F(x)=\nu_{\varphi}((-\infty,x])=\int_{- \infty}^{x} |\varphi(x)|^2\,dx,
\end{equation*}
and analogous:
\begin{equation*}
    G(k)=\nu_{\hat\varphi}((-\infty,k])=\int_{- \infty}^{k} |\hat\varphi(k)|^2\,dk.
\end{equation*}
There is a non-decreasing mapping $T:\mathbb{R}\rightarrow\mathbb{R}$ such that
\begin{equation}
     F(x) = G(T(x)),\forall x\in \mathbb{R},
\end{equation}
namely the {\sl increasing rearrangement} $T(x)=(G^{-1}\circ F)(x)$, with 
$G^{-1}(s)=\inf\{t\in\mathbb{R}: G(t)>s\}$.
Consequently, we get an ordinary differential equation in $L^1(\mathbb{R})$ when differentiating both sides:
\begin{equation}\label{deq2}
      |\varphi(x)|^2 = |\hat\varphi(T(x))|^2 \, \frac{dT}{dx}(x).
\end{equation}
Of course, $T$ is not necessarily differentiable on $\mathbb{R}$, but we have
Lebesgue's theorem for the differentiability of monotone functions and
absolute continuity of the distribution functions at hand, though regularity
questions are not in focus here. Note that (\ref{deq2}) coincides with (\ref{MAEQ}) 
for $n=1$.

If $T(x)=x$, then  $|\varphi(x)|^2 = |\hat\varphi(x)|^2$, i.e. this a solution
for any Hermite function $m=0,1,2,\ldots$
\begin{equation}
    \psi_m(x) = \frac{H_m(x)}{\sqrt{2^m\,m!\,\pi^\frac{1}{2}}} e^{-\frac{1}{2}\,x^2},
\end{equation}
because $\hat\psi_m(k)=(-i)^m\,\psi_m(k)$. We conclude 
\begin{displaymath}
   \int_{\RR 2} \HH(x,k) d\gamma(x,k) = \int_{\mathbb{R}} \HH(x,x) |\psi_m(x)|^2\,dx, 
\end{displaymath}
for $\gamma=(Id\times T)_\#\nu_{\psi_m}$. On the other hand - note the difference, 
\begin{displaymath}
    \int_{\RR 2} \HH(x,k) d\mu_{\psi_m}(x,k) = \int_{\RR 2} \HH(x,k) |\psi_m(x)|^2\,
    |\psi_m(k)|^2 \,dx\,dk.
\end{displaymath}
For $\HH(x,k)=|x|^2\,|k|^2$ the integral in the first case evaluates for $m=0$
to $\frac{3}{4}$ and in the second case to $\frac{1}{4}$ (this is optimal). If
$\HH$ is additive ($a(k)+b(x)$), however, then both integrals coincide. 

Generally, we can solve 
\begin{equation}
    g'(x) = \frac{|\psi(x)|^2}{|\hat\psi(g(x))|^2}
\end{equation} 
for various known ground states $\psi$ of the Schrödinger equation, then we have
\begin{displaymath}
     E_0=E_\HH(\psi) = \int_{\mathbb{R}} \left(g(x)^2+\VV(x)\right)\, |\psi(x)|^2 dx.
\end{displaymath}
The convex function $S(x)$ whose existence was claimed in (A.$\ref{mcann_thm}$) is
determined by
\begin{displaymath}
     \frac{dS}{dx}(x) = g(x),
\end{displaymath} 
but there seems to be no connection with the action $S$ in the Hamilton-Jacobi 
equation nor with $p(x)=-\frac{\psi'(x)}{\psi(x)}$, as was to be expected when
recollecting that $\nu_\psi\otimes\nu_{\hat\psi}$ is optimal. Regarding $p(x)$
(which is a vector field if $n>1$), consider the differential equation
\begin{equation*}
   \varphi'(x)+p(x)\,\varphi(x)=0
\end{equation*}
and assume $\varphi\in C^2(\mathbb{R})$, then 
\begin{equation*}
   \varphi''(x)+(p'(x)-p(x)^2)\,\varphi(x)=0,
\end{equation*}
so that when setting
\begin{displaymath}
  p'(x) = p(x)^2 - \VV(x) + \lambda,
\end{displaymath}
we can generate potentials $V$ and ground states of the form
\begin{displaymath}
      \varphi_p(x) = C \, e^{-\int^x p(s)\,ds},
\end{displaymath}
which satisfy: $-\varphi''(x)+\VV(x)\,\varphi(x)=\lambda\,\varphi(x)$. Take
$p(x)=\sinh(x)$ for instance, yielding $\VV(x)=\cosh(x)\,[\cosh(x)-1]$, $\lambda=1$,
and
\begin{displaymath}
    \varphi_0(x)=C\,e^{-\cosh(x)}.
\end{displaymath}
The computation of the Fourier transform is often manageable (\cite{tables_trafos}), 
so that one can find $g(x)$ and check against $p(x)$. As mentioned earlier, it seems 
that unless $p(x)=x$, $p$ is not a transport map, i.e. not admissible.  
Nevertheless, it holds true that
\begin{equation*}
    \int_{\mathbb{R}} p(x)^2\,|\varphi_0(x)|^2\, dx =
    \int_{\mathbb{R}} g(x)^2\,|\varphi_0(x)|^2\, dx =
    \int_{\mathbb{R}} |\varphi_0(x)'|^2\,dx.
\end{equation*}

\subsection{Discrete case}
Concluding these expositions, we want to point out that the discrete case of 
the transportation problem leads to bi-stochastic matrices $m\in \mathcal{B}_n$, so 
that the Kantorovich problem reads (assuming all points have the same {\sl mass})
\begin{displaymath}
  \inf\left\{ \frac{1}{n}\sum_{i,j} \HH(x_i,k_j) \, m_{ij}: m\in \mathcal{B}_n \right\},
\end{displaymath} 
which is a linear minimization problem on the bounded convex subset $\mathcal{B}_n$
of all real $n\times n$ matrices. The analogy to 
\begin{displaymath}
   \int_{\Gamma_n} \HH(x,k)\, d\mu_\varphi,
\end{displaymath}
however, is
\begin{equation}\label{discreteH}
     \sum_{i,j} \HH(x_i,k_j) |\varphi(x_i)|^2\, |\hat\varphi(k_j)|^2 \Delta x\,\Delta k,
\end{equation}
where the (normalized) complex vectors 
\begin{displaymath}
(\varphi(x_1),\ldots,\varphi(x_n))\ \mathrm{and}\ 
(\hat\varphi(k_1),\ldots,\hat\varphi(k_n)) 
\end{displaymath}
are related by the {\sl discrete (unitary)
Fourier} transformation. This leads to interesting problems  
that are slightly harder but are also relatively easy accessible from the numerical 
viewpoint due to the {\sl Fast Fourier Transform} algorithm. 
The standard Hamiltonian $\HH(x,k)=\frac{\hbar^2}{2m} k^2+\VV(x)$ even admits a 
closed form for the sum in (\ref{discreteH}), as is outlined in the two papers \cite{BalintGrid2} and
\cite{BalintGrid1}:
\begin{equation}\label{balintEq}
\HH_{ij}^0 = \begin{cases} 
  \frac{h^2}{4 m L^2}\left[\frac{N^2+2}{6}\right]+V_i 
    &\mbox{if } i=j \\
    (-1)^{(i-j)} \frac{h^2}{4 m L^2} \, \left[\sin(\pi \frac{(i-j)}{N})\right]^{-2}
    & \mbox{if } i\neq j \end{cases}, 
\end{equation}
where $\HH_{ij}^0 = \HH_{ij} \Delta x$, $L=N \Delta x$ and $\Delta k=\frac{2\pi}{L}$.
This provides a fast method to compute eigenvalues and eigenvectors, just computing
the potential at the grid points ($V_i$) then diagonalizing. There is an implementation
in {\tt FORTRAN77} by the authors
\begin{verbatim}
FGHEVEN 
C      F. Gogtas, G.G. Balint-Kurti and C.C. Marston,
C      QCPE Program No. 647, (1993).
\end{verbatim}
{\small\tt http://www.chm.bris.ac.uk/pt/dixon/dynamics/fghqcpe.for} \\
where it is also stated:
\begin{verbatim}
C  The analytical Hamiltonian expression given in the reference
C  contains a small error.  The formula should read: 
C
C    H(i,j) = {(h**2)/(4*m*(L**2)} *
C              [(N-1)(N-2)/6 + (N/2)] + V(Xi),   if i=j
C
C    H(i,j) = {[(-1)**(i-j)] / m } *
C              { h/[2*L*sin(pi*(i-j)/N)]}**2 ,     if i#j
\end{verbatim}
So, formula (\ref{balintEq}) should be correct, we only simplified the 
expressions between the square brackets.

Generally, we can easily translate to discrete spaces as follows: let 
$\varphi\in {\CC N}$ be a unit vector, and denote by $\hat\varphi\in {\CC N}$
the unitary DFT: $\hat\varphi = \frac{1}{\sqrt{N}} {\bf F} \varphi$, where
${\bf F}_{jk}=\omega_N^{(j-1)\,k}$ is the Vandermonde matrix, and 
$\omega_N = e^{-2\pi i/N}$ (a $N$-th root of unity). Then the measures
corresponding to $\nu_\varphi$ and $\nu_{\hat\varphi}$ are the linear
functionals (we identify the dual space with $\CC N$ here)
\begin{equation}
     \nu_\varphi = \sum_{i=1}^N |\varphi_i|^2\,\delta_i,\ \ \ 
     \nu_{\hat\varphi} = \sum_{i=1}^N |\hat\varphi_i|^2\,\delta_i,
\end{equation}
where $\delta_i(f)=f_i$ denotes the discrete Dirac measure (equivalently, the
dual base). Then $\Gamma(\varphi)$ corresponds to the matrices 
\begin{displaymath}
  \left\{ \gamma\in \mathbb{M}_{\mathbb{R}_{+,0}}^{N\times N}: 
    \sum_{i=1}^N(\gamma f)_i=\nu_{\varphi}(f),\ 
     \sum_{j=1}^N(\gamma^T f)_j=\nu_{\hat\varphi}(f),\ 
     \forall f\in {\CC N} \right\},
\end{displaymath}
which has at least the member
\begin{displaymath}
    \gamma_0=\nu_{\varphi}\otimes\nu_{\hat\varphi}=
     \sum_{i,j=1}^N |\varphi_i|^2\,|\hat\varphi_j|^2\,\delta_i\otimes \delta_j.
\end{displaymath}
The optimal transport problem between $\nu_\varphi$ and $\nu_{\hat\varphi}$ 
(roughly) reduces therefore to the {\sl linear program}
\begin{equation}
     \inf_{\gamma\in\Gamma(\varphi)} \operatorname{Tr}({\HH}\,\gamma^T)
\end{equation}
which has a solution, of course. Note that the (real) Frobenius scalar product 
$\langle \HH,\gamma\rangle_F$ which is often used in the literature is equal to 
the trace of $\HH \gamma^T$. Usually, an optimal $\gamma$ is a sparse matrix 
where at most $2N-1$ entries are different from zero, therefore, $\gamma_0$
generally is far from optimal. See the appendix for more details. 
\newpage
\begin{appendices}
\section{Examples and Remarks}
We cite here a version formulated by R.\,McCann
\cite{mccann_existence_1995}: 
\begin{theorem}[McCann]\label{mcann_thm}
Let $\mu,\nu\in\mathcal{P}({\RR n})$ and suppose $\mu$ vanishes on (Borel) subsets of $\RR n$ having Hausdorff dimension $n-1$. Then there exists a convex function $\psi$ on $\RR n$ whose gradient $\nabla\psi$ pushes $\mu$ forward to $\nu$. Although $\psi$ may not be unique, the map $\nabla\psi$ is uniquely determined $\mu$-almost everywhere.
\end{theorem}  
When we apply this result to the measures defined, we get 
the relation
\begin{equation}
   \int_{\RR n} f(\nabla\psi(x)) |\varphi(x)|^2 \, dx = \int_{\RR n} f(k) 
     |\hat\varphi(k)|^2 \, dk,
\end{equation}
where $\psi$ is the convex function whose existence was asserted. 
If the measures $\mu,\nu$ in the theorem are given by densities $\rho_0,\rho_1$ 
w.r.t. Lebesgue measure then 
\begin{displaymath}
     \int_{\RR n} f(\nabla\psi(x)) \rho_0(x) dx =
     \int_{\RR n} f(y) \rho_1(y) dy 
\end{displaymath}
implies (assuming sufficient regularity, see \cite{mccann_existence_1995} and
references therein) that $\psi$ is a solution to a {\sl Monge-Ampère} equation:
\begin{equation}
    \rho_0(x) = \rho_1(\nabla\psi(x)) \, \det D^2\psi(x).
\end{equation}
for instance,
\begin{equation}
     |\varphi(x)|^2 = |\hat\varphi(\nabla\psi(x))|^2 \, \det D^2\psi(x),
\end{equation}  
which represents a quite remarkable identity. For, if $\varphi$ has compact 
support in $\RR n$ then the Fourier transform $\hat\varphi$ is an analytic
function such that 
\begin{displaymath}
   |\hat\varphi(k)|^2 \leq \frac{|\mathtt{supp}(\varphi)|}{(2\pi)^n}
\end{displaymath}   
uniformly on $\RR n$, yielding the bound
\begin{displaymath}
   \det D^2\psi(x) \geq \frac{(2\pi)^n}{|\mathtt{supp}(\varphi)|}\,|\varphi(x)|^2
\end{displaymath}
On the other hand we recognize that outside $\mathtt{supp}(\varphi)$ either 
$\det D^2\psi(x)=0$ or $\nabla\psi$ must map into zeroes of $\hat\varphi$.
Of course, all these derivations need careful regularity analysis and indeed
most relations hold usually in a weak sense only. If we, however, assume that
$\psi\in C^2(\RR n)$ (for simplicity), then the inequality
\begin{equation}\label{detineq}
       \frac{1}{n} \, \Delta \psi(x) \geq (\det D^2\psi(x))^\frac{1}{n}
\end{equation}  
is an easy consequence of the convexity of $\psi$ and the well known inequality
$x_1+\ldots+x_n\geq n\,\sqrt[n]{x_1\cdot\ldots\cdot x_n}$, valid for all 
non-negative real numbers. Thus we obtain a bound in terms of the Laplacian of
$\psi$: 
\begin{equation}
      \Delta \psi(x) \geq  \,
      \frac{2\pi n}{|\mathtt{supp}(\varphi)|^\frac{1}{n}}\,|\varphi(x)|^\frac{2}{n}
\end{equation}
We will use the following simple lemma in the sequel:
\begin{lemma}
Let $f\in W^{1,\infty}(\RR n)$, $f(0)=0$ and $X\in L^2_{\mathtt{loc}}
(\RR n,\RR n)$ such that 
$\mathtt{div}(X)$ exists and is in $L^\infty_{\mathtt{loc}}(\RR n)$, then for all 
$\phi\in C^1_0(\RR n)$:
\begin{equation}\label{ineq}
     \int_{\RR n} |\nabla\phi(x)|^2\,dx \; 
     \int_{\RR n} |f'(\phi)|^2\, |X(x)|^2 \, dx \geq    
     \left(\int_{\RR n} f(\phi)\, \mathtt{div}(X) \, dx\right)^2
\end{equation}
with equality iff $\nabla\phi(x)+f'(\phi(x))\,X(x)=0$ almost everywhere. In that
case also $\norm{\nabla\phi}^2 = \norm{f'(\phi) X}^2 = |\int f(\phi)\,\mathtt{div}(X)\,dx|$.    
\end{lemma}
We only sketch the easy proof: expanding $\norm{\nabla\phi+f'(\phi)\,X}^2\geq 0$,
yields $\norm{\nabla\phi}^2+\norm{f'(\phi) X}^2 + \langle \nabla f(\phi),X\rangle_{L^2} \geq 0$. The last term may be rewritten to 
$\int \mathtt{div}(f(\phi)\,X)\,dx - \int f(\phi) \mathtt{div}(X)\,dx$. The
first term vanishes by the condition $f(0)=0$ so that $(\ref{ineq})$ follows
by optimizing the quadratic inequality w.r.t. $X$. It is also seen during the
proof that the inequality holds under much weaker conditions and that the 
boundary term may be incorporated if not zero. Replacing $\nabla\phi$ by 
$A\nabla\phi$, where $A$ is matrix function, also allows to consider degenerate
cases (if $A$ is indefinite). In fact, this is merely a disguised form of Schwarz's
inequality.

To illustrate the usefulness of $(\ref{detineq})$, let $X=\nabla\psi$ (the Brenier map to $\phi=\varphi$ from above), then using $(\ref{detineq})$, the right hand side of $(\ref{ineq})$ becomes:
\begin{equation}\label{estim1}
   \left(\int_{\RR n} f(\phi)\, \Delta\psi(x) \, dx\right)^2 \geq
    \frac{(2 \pi n)^2}{|\Omega|^\frac{2}{n}}\,
       \left( \int_{\Omega} f(\phi) |\phi(x)|^\frac{2}{n} \,dx, \right)^2
\end{equation}
where we have assumed $f\geq 0$ and used the abbreviation 
$\Omega=\mathtt{supp}(\phi)$.  
As a simple application let us consider the first Dirichlet eigenfunction $\phi_0$
with eigenvalue $\lambda_0(\Omega)$. Letting $f(t)=\frac{1}{2} t^2$ in 
$(\ref{ineq}),(\ref{estim1})$, and $X=\nabla\psi$ the corresponding Brenier map, 
we get
\begin{equation}\label{dirichlet}
     \int_{\Omega} |\nabla\phi_0(x)|^2\,dx \; 
     \int_{\Omega} |\phi_0|^2\, |\nabla\psi(x)|^2 \, dx \geq    
      \frac{(2 \pi n)^2}{4 |\Omega|^\frac{2}{n}}\,
       \left( \int_{\Omega} |\phi_0(x)|^{2+\frac{2}{n}} \,dx\right)^2
\end{equation}
By definition of $\lambda_0$ and $\psi$ we have 
\begin{displaymath}
    \lambda_0(\Omega)=\int_{\Omega} |\nabla\phi_0(x)|^2\,dx \;  =
     \int_{\Omega} |\phi_0|^2\, |\nabla\psi(x)|^2 \, dx ,
\end{displaymath}
so that $(\ref{dirichlet})$ gives
\begin{displaymath}
 \lambda_0(\Omega) \geq \frac{n \pi}{|\Omega|^\frac{1}{n}} 
    \, \norm{\phi_0}_{2+2/n}^{2+2/n}.
\end{displaymath}
Finally, recalling that generally $|\Omega|^{1-q/p}\,\norm{u}_p^q\leq \norm{u}_q^q$ holds for $p\leq q$, we get with $q=2+2/n$ and $p=2$ and because $\norm{\phi_0}$=1:    
\begin{equation}\label{lbound}
        \lambda_0(\Omega) \geq \frac{n \pi}{|\Omega|^\frac{2}{n}}. 
\end{equation}
This is apparently far from the optimal result stated by the well known 
{\sl Rayleigh-Faber-Krahn} inequality which gives the optimal value 
($B_\Omega=$ball with volume $|\Omega|$):
\begin{displaymath}
    \lambda_0(\Omega)\geq \lambda_0(B_\Omega)= 
      \left(\frac{\omega_n}{|\Omega|}\right)^\frac{2}{n}\, j_{\frac{n}{2}-1,1}^2.
\end{displaymath}
Indeed, the first zero of the Bessel function $J_m$ behaves like 
$j_{d,1} \sim d+O(d^\frac{1}{3})$ for large $d$, so $(\ref{lbound})$ not even
shows the correct asymptotics. Let $Q_{n,L}$ be the cube with side length $L$ in
$\RR n$, then $\lambda_0(Q_{n,L})=n \frac{\pi^2}{L^2}$ and
\begin{displaymath}
     Q_{n,{2/\sqrt{n}}} \subset B_1 \subset Q_{n,2}
\end{displaymath}   
so that by the domain monotonicity property of the eigenvalues
\begin{displaymath}
        n^2\,\frac{\pi^2}{4}\geq \lambda_0(B_1) \geq n\,\frac{\pi^2}{4}
\end{displaymath}
holds. Of course, the estimate $|\hat\varphi(k)|^2\leq\frac{|\Omega|}{2\pi}$ seems
rather crude, yet one would expect a sharper bound for large orders in view of the
fact that the same (at least similar) method provides a proof of the isoperimetric
inequality.     
\section{Discretization and Numerics  (\tmverbatim{UCDFT})}
Given a vector $X = (X_0, \ldots, X_{N - 1}) \in \mathbb{C}^n$, the
{\tmem{discrete Fourier transform}} of $X$ is defined as
\[ \mathtt{DFT} (X)_k = \sum_{j = 0}^{N - 1} X_j e^{- 2 \pi i \frac{jk}{N}}
\]
for $k = 0, \ldots, N - 1.$ The {\tmem{inverse discrete Fourier transform}} is
defined similarly:
\[ \mathtt{IDFT} (X)_j = \frac{1}{N} \sum_{k = 0}^{N - 1} X_k e^{2 \pi i
   \frac{kj}{N}}, \]
where$j = 0, \ldots, N - 1.$
\begin{prop}
  $\mathtt{\tmop{IDFT} \circ \tmop{DFT} = \tmop{Id}}$
\end{prop}
\begin{proof}
  \[ \mathtt{IDFT} \left( \mathtt{DFT} (X) \right)_j = \frac{1}{N} \sum_{k =
     0}^{N - 1} \left( \sum_{l = 0}^{N - 1} X_l e^{- 2 \pi i \frac{lk}{N}}
     \right) e^{2 \pi i \frac{kj}{N}} \]
  \[ = \frac{1}{N} \sum_{l = 0}^{N - 1} X_l  \sum_{k = 0}^{N - 1} e^{2 \pi i
     \frac{k (j - l)}{N}} = \frac{1}{N} \sum_{l = 0}^{N - 1} X_l N \delta_{j
     l} = X_j . \]
\end{proof}
\begin{remark}
  By symmetry we also have $\mathtt{\tmop{DFT} \circ \tmop{IDFT} =
  \tmop{Id}}$. Note that \tmverbatim{DFT} is not unitary because of the
  asymmetry of the factor $\frac{1}{N}$ in \tmverbatim{IDFT}. This could be
  fixed by using $\frac{1}{\sqrt{N}}$ in both transformations, however, one
  has to be careful when using numeric packages:
\end{remark}
\[ \hat{X} \assign \frac{1}{\sqrt{N}} \mathtt{DFT} (X), \hspace{2em} \check{X}
   \assign \sqrt{N} \mathtt{IDFT} (X) . \]
For our purposes we are more interested in the {\sl Centered Discrete Fourier
Transform} ({\tt UCDFT}). That is why we start from scratch. For details we
refer to \cite{auslander1979},\cite{Grunbaum1982} and \cite{Mugler2011}. 
\subsection{Discretization of the Fourier Transform}
Let us recall the definition of the unitary Fourier transformation on
$\mathbb{R}$:
\begin{definition}
  Let $f \in L^1 (\mathbb{R})$.
  \[ \hat{f} (k) = \frac{1}{\sqrt{2 \pi}} \int_{- \infty}^{\infty} f (x) e^{-
     ikx} \: d x \]
\end{definition}
An obvious approach to discretize $\hat{f}$ is
\[ \hat{f}_{L, M} (k) = \frac{1}{\sqrt{2 \pi}}  \sum_{m = - M}^M f (m \Delta
   x) e^{- imk \Delta x} \Delta x, \]
where $\Delta x = \frac{L}{2 M + 1}$, and $f$ is assumed to be continuous from
now on. We are here not interested in the quality of such approximations, only
in the connection to the \tmverbatim{DFT.} To recover $f$ from
$\widehat{f_{}}_{L, M}$ at the points $x_m = m \Delta x$, we have to evaluate
$\widehat{f_{}}_{L, M}$ at the points
\[ k_n = n \Delta k, \hspace{2em} \tmop{where} \hspace{1em} \Delta k = \frac{2
   \pi}{L} = \frac{2 \pi}{(2 M + 1) \Delta x}, \]
where $n = - M, \ldots, M.$
\
\begin{definition}
  Let $x_m = m \Delta x$, $k_n = n \Delta k$, where $n, m = - M, \ldots, M$,
  and $L = (2 M + 1) \Delta x.$
  \[ \hat{f}_{L, M} (k_n) : = \frac{1}{\sqrt{2 \pi}}  \sum_{m = - M}^M f (x_m)
     e^{- ix_m k_n} \Delta x. \]
\end{definition}
\begin{prop}
  With the definitions above we can recover $f$ at the points $x_r = r \Delta
  x$ by
  \[ f (x_r) = \frac{1}{\sqrt{2 \pi}}  \sum_{n = - M}^M \hat{f}_{L, M} (k_n)
     e^{ix_r k_n} \Delta k. \]
  $r = - M, \ldots, M.$
\end{prop}
\begin{proof}
  \[ \frac{1}{\sqrt{2 \pi}}  \sum_{n = - M}^M \left( \frac{1}{\sqrt{2 \pi}} 
     \sum_{m = - M}^M f (x_m) e^{- ix_m k_n} \Delta x \right) e^{ix_r k_n}
     \Delta k = \]
  \[ \frac{1}{2 \pi} \sum_{n = - M}^M \sum_{m = - M}^M f (x_m) e^{- ix_m k_n}
     e^{ix_r k_n} \Delta x \Delta k = \]
  \[ \frac{1}{2 \pi} \sum_{n = - M}^M \sum_{m = - M}^M f (x_m) e^{i (x_r -
     x_m) k_n} \Delta x \Delta k = \]
  \[ \frac{1}{2 \pi} \sum_{m = - M}^M f (x_m)  \left( \sum_{n = - M}^M e^{i
     (x_r - x_m) k_n} \right) \Delta x \Delta k = \]
  \[ \  \]
  \[ \frac{2 M + 1}{2 \pi} \sum_{m = - M}^M f (x_m) \delta_{r m} \Delta x
     \Delta k = \frac{2 M + 1}{2 \pi} f (x_r) \Delta x \Delta k =
      f (x_r) \]
\end{proof}
\
Now, set $N = 2 M + 1.$
\begin{definition}
  We define the vectors 
  \[X = (X_0, \ldots, X_{N - 1}), Y = (Y_0, \ldots, Y_{N
  - 1}) \in \mathbb{C}^N
  \] 
  as follows:
  \[ X_m = f ((m - M) \Delta x) \Delta x, \hspace{1em} m = 0, \ldots, 2 M. \]
  \[ Y_n = \hat{f}_{L, M} ((n - M) \Delta k) \Delta k, \hspace{1em} n = 0,
     \ldots, 2 M. \]
  Note that $2 M = N - 1,$and $\Delta x \Delta k = \frac{2 \pi}{(2 M + 1)} =
  \frac{2 \pi}{N} .$
\end{definition}
\
Therefore, we get from
\[ \hat{f}_{L, M} (k_n) = \frac{\Delta x}{\sqrt{2 \pi}}  \sum_{m = - M}^M f (m
   \Delta x) e^{- imn \Delta k \Delta x}, \hspace{1em} n = - M, \ldots, M, \]
by inserting $X, Y :$
\[ Y_{n + M} = \frac{\Delta k}{\sqrt{2 \pi}}  \sum_{m = - M}^M X_{m + M} e^{-
   2 \pi i \frac{mn}{2 M + 1}} . \]
Now, let $m' = M + m$, $n' = M + n$, then
\[ Y_{n'} = \frac{\Delta k}{\sqrt{2 \pi}} e^{2 \pi i \frac{M (n' - M)}{N}} 
   \sum_{m' = 0}^{N - 1} X_{m'} e^{- 2 \pi i \frac{m'  (n' - M)}{N}} \]
where now $m', n' = 0, \ldots, N - 1.$ Thus we finally obtain
\[ Y_{n'} = \frac{\Delta k}{\sqrt{2 \pi}} e^{2 \pi i \frac{M (n' - M)}{N}} 
   \mathtt{DFT} \left( \left\{ X_{m'} e^{2 \pi i \frac{m' M}{N}} \right\}
   \right)_{n'} . \]
\subsubsection{Matrix representation}\label{MATREP}
Let $N = 2 M + 1, \omega_N = e^{- 2 \pi i / N}$, then we define the $N \times
N$ matrix
\[ U_{m n} = \frac{1}{\sqrt{N}} \omega_N^{(m - M - 1)  (n - M - 1)},
   \hspace{1em} m, n = 1, \ldots, N. \]
Therefore,
\[ \hat{f}_{L, M} (k_{n - M - 1}) : = \frac{\Delta x \sqrt{N}}{\sqrt{2 \pi}} 
   \sum_{m = 1}^N U_{m n} f (x_{m - M - 1}) = \frac{L}{\sqrt{2 \pi N}} 
   \sum_{m = 1}^N U_{m n} f (x_{m - M - 1}) . \]
The inverse is given by
\[ f (x_{r - M - 1}) = \frac{\sqrt{N}}{\sqrt{2 \pi}}  \sum_{n = 1}^N U_{r
   n}^{\star} \hat{f}_{L, M} (k_{n - M - 1}) \Delta k = \frac{\sqrt{2 \pi
   N}}{L}  \sum_{n = 1}^N U_{r n}^{\star} \hat{f}_{L, M} (k_{n - M - 1}) . \]
\
The matrix $U$ is unitary.
\[ \hat{f}_{L, M} = \frac{L}{\sqrt{2 \pi N}} U f \Rightarrow | \hat{f}_{L, M}
   |^2 = \frac{L^2}{2 \pi N} | f |^2 . \]
Note that $| \cdot |$ means the vector norm and not the $\| \cdot \|_2$ norm.
The latter is the former times $\Delta x$ or $\Delta k$ respectively. See
below.
\
\subsubsection{Normalization condition}
The corresponding discrete expression to
\[ \int_{\mathbb{R}} | f (x) |^2 d x \]
is
\[ \sum_{m = - M}^M | f (x_m) |^2 \Delta x = 1. \]
Thus,
\[ \sum_{m = 1}^N | f (x_{m - M - 1}) |^2 = \frac{1}{\Delta x} =
   \frac{N_{}}{L} \Rightarrow \| v_f \| = \sqrt{\frac{N}{L} .} \]
In $k$-space, however,
\[ \sum_{m = 1}^N | g (k_{m - M - 1}) |^2 = \frac{1}{\Delta k} = \frac{L_{}}{2
   \pi} \Rightarrow \| v_g \| = \sqrt{\frac{L}{2 \pi} .} \]
\subsubsection{Grid space calibration}
Recalling the grid spacing in $x$ and $k$ space
\[ \Delta x = \frac{L}{N}, \hspace{1em} \Delta k = \frac{2 \pi}{L}, \]
we see that they usually are different. The condition for $\Delta x = \Delta
k$ is
\[ L^2 = 2 \pi N, \]
thus, the factors in the matrix representation above will become
$\frac{L}{\sqrt{2 \pi N}} = 1$.
\
\
\subsection{The Measures}
\subsubsection{The projection measures (marginals)}
The discrete analogues to $\nu_{\varphi}$ and $\nu_{\hat{\varphi}}$ are
\[ \nu_{\varphi} = \Delta x \sum_{m = - M}^M | \varphi (x_m) |^2 \delta_{x_m},
   \hspace{1em} \nu_{\hat{\varphi}} = \Delta k \sum_{m = - M}^M |
   \hat{\varphi} (k_m) |^2 \delta_{k_m} . \]
\subsubsection{The couplings}
The admissible set of measures is (as in the continuous case):
\[ (\pi_{1, \#} \gamma) (f) = \gamma (\pi_1^{\#} f) = \gamma (f \circ \pi_1) = \]
\[
   \Delta x \Delta k \sum \gamma_{m, n} f (\pi_1 (x_m, k_n)) = \Delta x \sum |
   \varphi (x_m) |^2 f (x_m) . \]
Therefore,
\[ \Gamma (\varphi) = \left\{ \gamma \in \mathbb{R}_+^{N \times N} : \sum_{n
   = - M}^M \gamma_{m n} = | \varphi (x_m) |^2, \sum_{m = - M}^M \gamma_{m n}
   = | \hat{\varphi} (k_n) |^2 \right\} \Delta x \Delta k. \]
\
\subsubsection{Push forward}
\[ T_{\#} \nu_{\varphi} (f) = \nu_{\varphi} (f \circ T) = \]
\[\Delta x \sum_{m = -
   M}^M | \varphi (x_m) |^2 \delta_{x_m} (f \circ T) = \Delta x \sum_{m = -
   M}^M | \varphi (x_m) |^2 f (T (x_m)) \]
\[ T_{\#} \nu_{\varphi} (f) = \nu_{\hat{\varphi}} (f) = \Delta k \sum_{m = -
   M}^M | \hat{\varphi} (k_m) |^2 f (k_m) \]
\[ \sum_{m = - M}^M | \varphi (x_m) |^2 f (T (x_m)) \Delta x = \sum_{m = -
   M}^M | \hat{\varphi} (k_m) |^2 f (k_m) \Delta k \]
\[ \sum_{m = - M}^M | \varphi (T^{- 1} (y_m)) |^2 f (y_m)  \frac{\Delta y}{T'
   (T^{- 1} (y_m))} = \sum_{m = - M}^M | \hat{\varphi} (k_m) |^2 f (k_m)
   \Delta k \]
\[ x_m = T^{- 1} (y_m) \Rightarrow \Delta y = T (x_{m + 1}) - T (x_m) = T'
   (x_m) \Delta x \]
\subsection{A discrete model}
\subsubsection{The energies}
\[ E_H (\varphi) = \sum_{m, n = - M}^M H (x_m, k_n)  | \varphi (x_m) |^2 |
   \hat{\varphi} (k_n) |^2 \Delta x \Delta k \]
\[ K_H (\varphi) = \inf_{\gamma \in \Gamma (\varphi)} \sum_{m, n = - M}^M H
   (x_m, k_n) \gamma_{m n} \Delta x \Delta k \]
\[ M_H (\varphi) = \inf_{T_{\#} \nu_{\varphi} = \nu_{\hat{\varphi}}}  \sum_{m
   = - M}^M H (x_m, T (x_m))  | \varphi (x_m) |^2 \Delta x \]
\subsubsection{Ground state of $E_H$}
With $X_m = \varphi (x_{m - M - 1}) \sqrt{\Delta x}$, $m = 1 \ldots N = 2 M +
1, \lambda_0 = \inf_{\varphi} E_H (\varphi)$ becomes
\[ \lambda_0 (H) = \inf_{| X | = 1} \sum_{1 \leqslant i, j \leqslant N} H_{i
   j}  | X_i |^2  | (U X)_j |^2, \]
where $U$ is the matrix given in \ref{MATREP}, and
\[ H_{i j} = H (x_{i - M - 1}, k_{j - M - 1}) . \]
The normalization condition $| X | = 1$ corresponds to $\| \varphi \|_2 = 1$.
Indeed,
\[ \sum_{m = 1}^N | X_m |^2 = \sum_{m = 1}^N | \varphi (x_{m - M - 1}) |^2
   \Delta x = \sum_{l = - M}^M | \varphi (x_l) |^2 \Delta x. \]
\subsubsection{Special cases: $H = A + B, H = AB$}
\[ \lambda_0 (A + B) = \inf_{| X | = 1} \left\{ \sum_i A_i  | X_i |^2 + \sum_j
   B_j | (U X)_j |^2 \right\} . \]
Consider
\[ \sum_j B_j | (U X)_j |^2 = \sum_{j, l, m} B_j U_{j l} X_l U^{\star}_{j m}
   X_m^{\star}, \]
therefore,
\[ \sum_{j, l, m} (A_j \delta_{j l} \delta_{j m} X_l X^{\star_{}}_m + B_j U_{j
   l} U_{j m}^{\star} X_l X^{\star_{}}_m) = \sum_{l, m} C_{l m} X_l
   X^{\star_{}}_m \]
where
\[ C_{l m} = \sum_j (A_j \delta_{j l} \delta_{j m} + B_j U_{j l} U_{j
   m}^{\star} ) = A_l \delta_{l m} + \sum_j B_j U_{j l} U_{j m}^{\star}  \]
Recalling the definition of $U :$
\[ C_{l m} = A_l \delta_{l m} + \frac{1}{N} \sum_j B_j \omega_N^{(j - M - 1) 
   (l - M - 1)} \omega_N^{- (j - M - 1)  (m - M - 1)} = \]
\[ A_l \delta_{l m} + \frac{1}{N} \sum_j B_j \omega_N^{(j - M - 1)  (l - m)} =
   A_l \delta_{l m} + \frac{\omega_N^{- (M + 1)  (l - m)}}{N} \sum_j B_j
   \omega_N^{j (l - m)} . \]
Since
\[ H (x, k) = A (x) + B (k) . \]
we get
\[ H_{l m} = A_l \delta_{l m} + \frac{\omega_N^{- (M + 1)  (l - m)}}{N}
   \sum_{j = 1}^N B_j \omega_N^{j (l - m)}, \]
where $A_l = A (x_{l - M - 1})$ and $B_m = B (k_{m - M - 1})$.
\
\subsubsection{Special case $B (k) \sim | k |^2$}
For this case we ought to obtain a similar expression as in (\ref{balintEq}), i.e. 
a closed form for the term
\[ \kappa_N = \frac{\omega_N^{- (M + 1)  (l - m)}}{N} \sum_{j = 1}^N B_j
   \omega_N^{j (l - m)}, \]
with $B_j = B (k_{j - M - 1}) = \alpha (j - M - 1)^2,$where $\alpha$ is some
real valued constant. Let $s = l - m$, and recall that $M = \frac{N - 1}{2}$,
then
\[ \kappa_N (s) = \alpha \frac{\omega_N^{- \frac{N + 1}{2} s}}{N} \sum_{j =
   1}^N \left( j - \frac{N + 1}{2} \right)^2 \omega_N^{j s} . \]
A short calculation yields for $s \neq 0 :$
\[ \kappa_N (s) = \alpha (- 1)^s \frac{\cos \left( \frac{\pi s}{N} \right)}{1
   - \cos \left( 2 \frac{\pi s}{N} \right)} = \alpha (- 1)^s
   \frac{\cos \left( \frac{\pi s}{N} \right) }{2 \sin \left( \frac{\pi s}{N}
   \right)^2}, \]
where we used the formulae
\[ S_N (q) = \sum_{j = 1}^N q^j = q^{}  \frac{1 - q^N}{1 - q^{}}, \hspace{1em}
   q S'_N (q) = \sum_{j = 1}^N jq^j, \hspace{1em} q (q S'_N (q))' = \sum_{j =
   1}^N j^2 q^j . \]
If $s = 0,$ we get
\[ \kappa_N (0) = \frac{\alpha}{N} \sum_{j = 1}^N \left( j - \frac{N + 1}{2}
   \right)^2 = \alpha \frac{N^2 - 1}{12} . \]
Thus, the kinetic term $\frac{\hbar^2}{2 m} | k |^2$ becomes
\[ T_{i j} = \frac{h^2}{8 \pi mN}  \left\{ \begin{array}{l}
     \frac{N^2 - 1}{6}, i = j\\
     (- 1)^{(l - m)} \frac{\cos \left( \frac{\pi (i - j)}{N}
     \right) }{\sin \left( \frac{\pi (i - j)}{N} \right)^2}, i \neq j
   \end{array} \right. . \]
With
   \[ \sum_{m = - M}^M V (x_m)  | \varphi (x_m) |^2 \Delta x = \sum_{m = -
      M}^M V (x_m)  | X_m |^2 = \]
   \[ \sum_{m = 1}^N V (x_{m - M + 1})  | X_m |^2 \Rightarrow V_m = V (x_{m -
      M + 1}) = V ((m - M + 1) \Delta x) . \] 
we get finally:
\[ H_{i j} = V_i \delta_{i j} + \frac{h^2}{8 \pi mN}  \left\{ \begin{array}{l}
     \frac{N^2 - 1}{6}, i = j\\
     (- 1)^{(i - j)}  \frac{\cos \left( \frac{\pi (i - j)}{N}
     \right) }{\sin \left( \frac{\pi (i - j)}{N} \right)^2}, i \neq j
   \end{array} \right. . \]
This is different from (\ref{balintEq}) because we use an odd number of grid 
points $(N = 2M + 1),$ whereas (\ref{balintEq}) will become 
(when inserting $L^2 = 2 \pi N$):
 \[ H_{i j}^0 = V_i \delta_{i j} + \frac{h^2}{8 \pi mN}  \left\{
   \begin{array}{l}
     \frac{N^2 + 2}{6}, i = j\\
     (- 1)^{(i - j)}  \frac{1}{\sin \left( \frac{\pi (i -
     j)}{N} \right)^2}, i \neq j
   \end{array} \right. . 
 \]
\subsection{Numerics}
Just for illustration purposes we calculated some quantities using {\sl GNU Octave}
\cite{octave}.
{\small
\begin{verbatim}
  John W. Eaton, David Bateman, Soren Hauberg, Rik Wehbring (2016).
  GNU Octave version 4.2.0 manual: a high-level interactive 
  language for numerical computations.
  URL http://www.gnu.org/software/octave/doc/interpreter/
\end{verbatim}
}
Due to the lack of analytical examples regarding the optimal transport measures
it will be be rather wishful to have some reliable numerical instances at least. 
Of course, the small script following does not claim either correctness nor completeness.



\section*{Supplementary material}

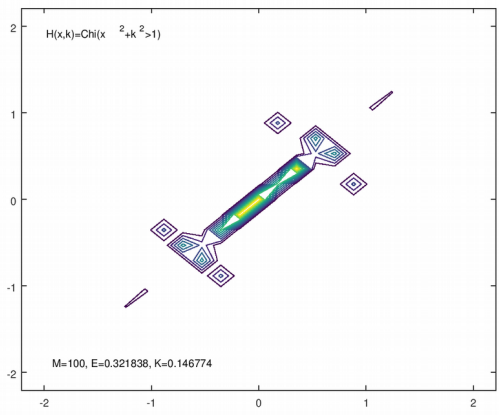
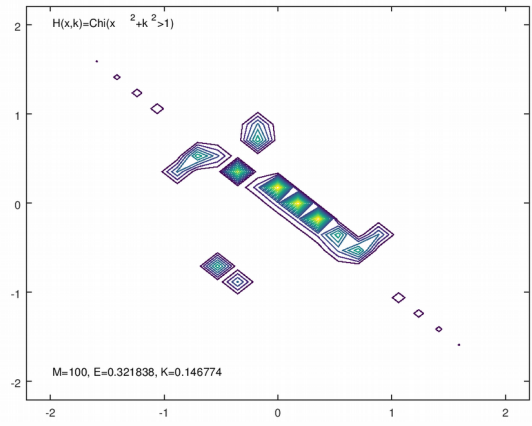
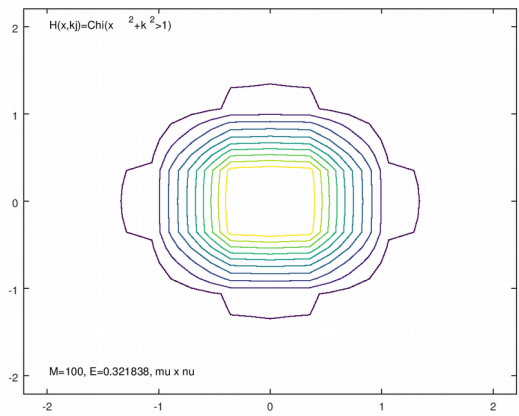
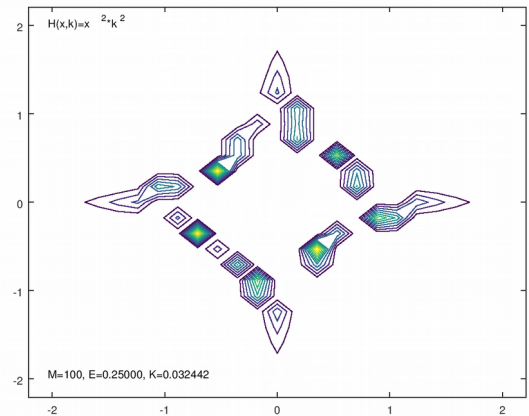
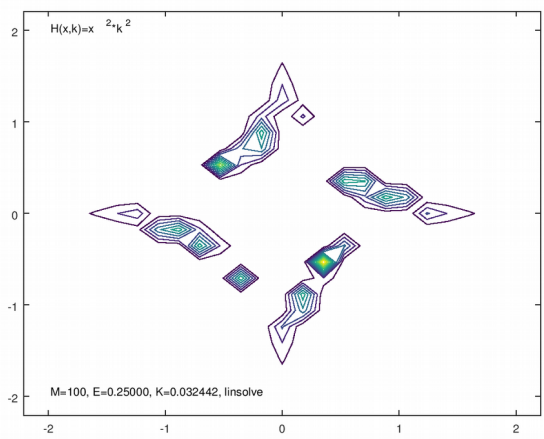
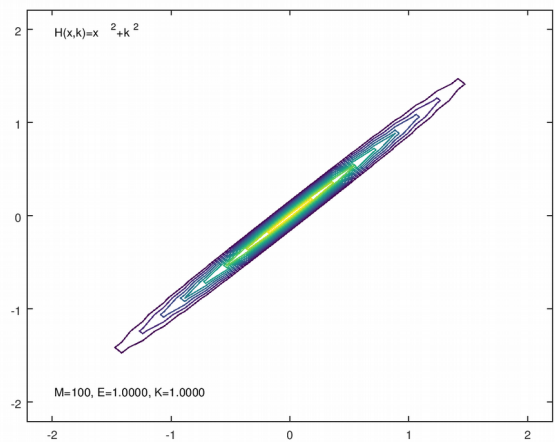

```octave
clear all;

M = 100;
H = @(x,k) x^2*k^2;

function [N,L,x,k,dx,dk] = setup_grid(M)
  N = 2*M+1;
  L = sqrt(2*pi*N);
  dx = L/N;
  dk = 2*pi/L;
  x = (-M:1:M)*dx; #x = linspace(-M,M,N)*dx
  k = (-M:1:M)*dk; #k = linspace(-M,M,N)*dk
endfunction


function rv = dimGetAndCheck(x,y)
  # Checks if x,y are vectors of the same size (col/row does not matter).
  #
  if ~isvector(x) error("x: vector expected.") endif
  if ~isvector(y) error("y: vector expected.") endif
  [x0,x1] = size(x);
  [y0,y1] = size(y);
  N = 1 + abs(x1-x0);
  if (N == 1 + abs(y1-y0))
    rv = N;
  else
    error("vectors of equal size required.");
  endif
endfunction


function rv = FunctionMatrix(f,x,y)
  # FunctionMatrix(f,x,y)= matrix f(x_i,y_j)
  #
  N = dimGetAndCheck(x,y);
  F = eye(N);
  for m=1:N
    for n=1:N
      F(m,n) = f(x(m),y(n));
    endfor
  endfor
  rv = F;
endfunction


function rv = UCDFT(M)
  # Unitary Centered Discrete Fourier Transform (UCDFT)
  # M = (N-1)/2 : integer
  # Computes the NxN matrix
  #    UCDFT(M)= 1/sqrt(N)*[exp(-2*i*pi/N)^((m-M-1)*(n-M-1))]
  #
  N = 2*M+1;
  U = eye(N);
  w = exp(-2*i*pi/N);
  for m=1:N
    for n=1:N
      U(m,n)=w^((m-M-1)*(n-M-1));
    endfor
  endfor
  rv = U/sqrt(N);
endfunction


function checkL2funvec(phi,M)
  # phi has to be a L^2 normalized complex valued row vector with
  # phi(1)=phi(N)=0.
  #
  eps0 = 1e-9;
  if ~isvector(phi) error("vector expected.") endif
  if ~isrow(phi) error("row vector expected.") endif
  if (abs(norm(phi)-1.0)>eps0) error("norm(vector,2) is not 1.") endif
  if (length(phi)~=2*M+1) error("size of phi must be 2*M+1.") endif
  if ~(phi(1)==0 && phi(end)==0) error("boundary condition not satisfied") endif
endfunction


function rv = makePhi(f,x)
  # Create a row vector that will pass checkL2funvec from a function f
  # Ex. phi=makePhi(@(t) t^2,[1,2,3,4,5]) -> checkL2funvec(phi,2)
  fx = arrayfun(f,x);
  fx(1)=0;
  fx(end)=0;
  rv = fx/norm(fx);
endfunction


function [mu,nu] = projMeasures(phi,M)
  # Create the marginal measures mu,nu from the function-vector phi.
  checkL2funvec(phi,M);
  dftphi = UCDFT(M) * phi';
  psi = dftphi' ;
  mu=phi .* conj(phi);   # sum(mu)=sum(nu)=1 ? check
```

```octave
  nu=psi .* conj(psi);
endfunction


function rv = tensorProductMeasure(mu,nu)
  # g = mu \otimes nu => g(i,j)=mu(i)*nu(j)
  # Assumes that all checks have been done before
  rv = mu' .* nu;
endfunction


function rv = makeConstraintMatrix(M)
  # Build a matrix A for LP
  # Use A = makeConstraintMatrix(M) -> sparse A
  # full(A) builds normal reprsentation.
  flat = @(x)x(:);
  mkc = @(n)sparse( flat(repmat(1:n, [n 1])), ...
           flat(reshape(1:n*n,n,n) ), ...
           ones(n*n,1) );
  mkr = @(n)sparse( flat(repmat(1:n, [n 1])), ...
           flat(reshape(1:n*n,n,n)' ), ...
           ones(n*n,1) );
  mkA = @(n)[mkr(n);mkc(n)];
  rv = mkA(2*M+1);
endfunction


function [gamma,mu,nu,x,k,EK,ES,HM,phi,FMIN,ERRNUM] = SolveWithGLPK(H,phi,M)
  [N,L,x,k,dx,dk] = setup_grid(M);
  [mu,nu] = projMeasures(phi,M);
  tpm = tensorProductMeasure(mu,nu);
  HM = FunctionMatrix(H,x,k);
  C = HM(:);
  A = makeConstraintMatrix(M);
  b = [mu,nu](:);
  [XOPT, FMIN, ERRNUM, EXTRA]=glpk(C,A,b);
  gamma = reshape(XOPT,N,N);
  EK=trace(HM'*gamma);
  ES=trace(HM'*tpm);
endfunction


function checkGLPKSolution(g,m,n,ek,err,fmin)
  # Checks the results from "glpk".
  eps0 = 1e-7;
  gm = norm(sum(g,2)-m(:));
  gn = norm(sum(g,1)'-n(:));
  gp = full(sum(g(:)~=0));
  de = abs(fmin-ek);
  N = length(m);
  if (err~=0) fprintf("Warning: ERRNUM = %d \n",err); endif
  if (de > eps0) fprintf("Warning: FMIN=%d <> EK \n",fmin); endif
  if (gm>eps0) fprintf("Warning: px(gamma)- mu = %d \n",gm); endif
  if (gn>eps0) fprintf("Warning: pk(gamma)- nu = %d \n",gn); endif
  if (gp>2*N-1) fprintf("Warning: #supp(gamma) = %d \n",gp); endif
endfunction


function [gamma,mu,nu,x,k,EK,ES,HM,phi] = Solve1(H,f,M)
  # Solves EK(phi)=min{g: Tr(H'*g), px(g)=mu(phi), pk(g)=nu(phi)}
  # using LP method/function glpk.
  [N,L,x,k,dx,dk] = setup_grid(M);
  phi = makePhi(f,x);
  [gamma,mu,nu,x,k,EK,ES,HM,phi,FMIN,ERRNUM] = SolveWithGLPK(H,phi,M);
  # checks
  checkGLPKSolution(gamma,mu,nu,EK,ERRNUM,FMIN);
endfunction


function [gamma,mu,nu,x,k,EK,ES,HM,phi,FMIN] = SolveWithLinprog(H,phi,M)
  # [XOPT,FMIN,ITN] = linprog(A,b,c,k,maxit,tol)
  # max iterations ~ 50k (10k is too small)
  maxit = 50000;
  tol = 1e-9;
  [N,L,x,k,dx,dk] = setup_grid(M);
  [mu,nu] = projMeasures(phi,M);
  tpm = tensorProductMeasure(mu,nu);
  HM = FunctionMatrix(H,x,k);
  C = HM(:);
  A = makeConstraintMatrix(M);
  b = [mu(:);nu(:)];
  [XOPT, FMIN, ITN]=perform_linprog(A,b,C,0,maxit,tol);
  gamma = reshape(XOPT,N,N);
  EK=trace(HM'*gamma);
  ES=trace(HM'*tpm);
endfunction

function [gamma,mu,nu,x,k,EK,ES,HM,phi] = Solve2(H,f,M)
  # Solves EK(phi)=min{g: Tr(H'*g), px(g)=mu(phi), pk(g)=nu(phi)}
  # using LP method/function perform_linprog.
  [N,L,x,k,dx,dk] = setup_grid(M);
  phi = makePhi(f,x);
```

```octave
  [gamma,mu,nu,x,k,EK,ES,HM,phi,FMIN] = SolveWithLinprog(H,phi,M);
  # checks
  checkGLPKSolution(gamma,mu,nu,EK,0,FMIN);
endfunction

####
function energy = ESchroed(HM,U,M,phi,x,k)
  X = phi'/norm(phi);
  Y = (U * X')';
  mu = X .* conj(X);
  nu = Y .* conj(Y);
  tpm = tensorProductMeasure(mu,nu);
  energy = trace(HM*tpm');
endfunction


function constraint = g(phi)
  [a,b] = size(phi);
  constraint = [sumsq(phi)-1; phi(a); phi(b)];
endfunction


function [phi0, obj, info, iter, nf, lambda,x] = findGroundState(H,M)
  [N,L,x,k,dx,dk] = setup_grid(M);
  HM = FunctionMatrix(H,x,k);
  U = UCDFT(M);
  E = @(phi) ESchroed(HM,U,M,phi,x,k);
  X0 = makePhi(@(t) exp(-t^2/2),x);
  [phi0, obj, info, iter, nf, lambda] = sqp (X0, E, @g, []);
endfunction


function [HM,phi,E] = schroedEq(M,V,m,h)
  # Solve the Schroedinger equation by the FGH method
  [N,L,x,k,dx,dk] = setup_grid(M);
  VV = arrayfun(V,x);
  fac = h^2/(8*pi*m*N);
  HM = eye(N);
  for m=1:N
    for n=1:N
      if (m==n)
        HM(m,n) = fac*(N^2-1)/6 + VV(m);
      else
        HM(m,n) = fac*(-1)^(m-n)*cos(pi*(m-n)/N)/sin(pi*(m-n)/N)^2;
      endif
    endfor
  endfor
  [vv,eed] = eig(HM);
  [E,perm] = sort(diag(eed));
  phi = vv(:,perm);
endfunction

# Examples:
# [gamma,mu,nu,x,k,EK,ES,HM,phi] = Solve1(@(x,k) x^2*k^2,@(t) exp(-t^2/2),M);
# [gamma,mu,nu,x,k,EK,ES,HM,phi] = Solve1(@(x,k) x^2+k^2,@(t) exp(-t^2/2),M);
# [gamma,mu,nu,x,k,EK,ES,HM,phi] = Solve2(@(x,k) x^2*k^2,@(t) exp(-t^2/2) ,M);


### Ground states
# [phi0, obj, info, iter, nf, lambda,x] = findGroundState(@(x,k) x^2+k^2,M);
# plot(x,phi0)

# [phi0, obj, info, iter, nf, lambda,x] = findGroundState(H,M);
# plot(x,phi0)

#  H=@(x,k) int8(x^2+k^2>4);
# [gamma1,mu,nu,x,k,EK1,ES,HM,phi] = Solve1(H,@(t) exp(-t^2/2) ,M);
# [gamma2,mu,nu,x,k,EK2,ES,HM,phi] = Solve2(H,@(t) exp(-t^2/2) ,M);


#[N,L,x,k,dx,dk] = setup_grid(M);
#V = @(t) t^2;
#[HM,phi,E] = schroedEq(M,V,1/2,2*pi);  # hbar^2=2*m => m=1/2, h=2*pi
#g0 = @(t) exp(-t^2/2);
#h0 = arrayfun(g0,x);
#h0 = h0/norm(h0);
#plot(x,-phi(:,1),x,h0)
```


\end{appendices}
\newpage
\bibliographystyle{plain}
\bibliography{omt}
\texttt{\tiny \$Id: transport.tex 3 2015-11-18 16:22:01Z pagani\$}\\
\texttt{\tiny Slightly extended, 18-MAR-2018.}\\
\texttt{\tiny Code will be updated @ github.com/nilqed/qmtrans \\}
\texttt{\tiny nilqed@gmail.com}
\end{document}